\documentclass[10pt, a4paper]{article}
\usepackage[body={15.4cm,23.6cm}, top=2.3cm]{geometry}

\usepackage{latexsym,bm,graphicx,url,cite,amsmath,amsthm,amssymb,amsfonts,mathrsfs,enumerate,subfigure,color,dsfont,epsfig,verbatim}
\usepackage{geometry} % to change the page dimensions
\usepackage{extarrows}
\usepackage[english]{babel}
\usepackage{setspace}
\usepackage[colorlinks,linkcolor=magenta,citecolor=magenta,anchorcolor=magenta]{hyperref}
\numberwithin{equation}{section}
\newtheorem{theorem}{Theorem}[section]
\newtheorem{definition}{Definition}[section]
\newtheorem{remark}{Remark}[section]
\newtheorem{proposition}{Proposition}[section]

\newtheorem{lemma}{Lemma}[section]

\newcommand{\Ex} {
\mathbb{E}}

\newcommand{\Px} {
\mathbb{P}}

\newcommand{\G} {
\mathcal{G}}
\newcommand{\R} {
\mathbb{R}}

\newcommand{\N}{\mathbb{N}}

\def\esssup_#1{\underset{#1}{\Xi}}
\def\essinf_#1{\underset{#1}{\mathrm{ess\,inf\, }}}
\def\argmax_#1{\underset{#1}{\mathrm{arg\,max\, }}}
\def\argmin_#1{\underset{#1}{\mathrm{arg\,min\, }}}

\newcommand{\Gx}{\mathbb{G}}
\newcommand{\Hx}{\mathbb{H}}
\newcommand{\Zx}{\mathbb{Z}}
\newcommand{\Fx}{\mathbb{F} }

\newcommand{\Gam}{\Gamma}

\definecolor{Red}{rgb}{1.00, 0.00, 0.00}

\definecolor{DRed}{rgb}{0.5, 0.00, 0.00}

\definecolor{Blue}{rgb}{0.00, 0.00, 1.00}

\definecolor{Green}{rgb}{0.0, 0.4, 0.0}

\definecolor{Magenta}{rgb}{1.0, 0, 1.0}

\allowdisplaybreaks
\title{}
\author{}
\title{Risk Sensitive Portfolio Optimization with Default Contagion and Regime-Switching}

\author{Lijun Bo \thanks{Email: lijunbo@ustc.edu.cn, School of Mathematical Sciences, University of Science and Technology of China, Hefei, Anhui
Province, 230026, China, and Wu Wen Tsun Key
Laboratory of Mathematics, Chinese Academy of Science, Hefei, Anhui Province 230026, China.}\and
Huafu Liao \thanks{Email: lhflhf@mail.ustc.edu.cn, School of Mathematical Sciences, University of Science and Technology of China, Hefei, Anhui
Province, 230026, China.} \and
Xiang Yu \thanks{E-mail: xiang.yu@polyu.edu.hk, Department of Applied Mathematics, The Hong Kong Polytechnic University, Hung Hom, Kowloon, Hong Kong.}}

\begin{document}

\maketitle

\begin{abstract}

We study an open problem of risk-sensitive portfolio allocation in a regime-switching credit market with default contagion. The state space of the Markovian regime-switching process is assumed to be a countably infinite set. To characterize the value function, we investigate the corresponding recursive infinite-dimensional nonlinear dynamical programming equations (DPEs) based on default states. We propose to work in the following procedure: Applying the theory of monotone dynamical system, we first establish the existence and uniqueness of classical solutions to the recursive DPEs by a truncation argument in the finite state space. The associated optimal feedback strategy is characterized by developing a rigorous verification theorem. Building upon results in the first stage, we construct a sequence of approximating risk sensitive control problems with finite states and prove that the resulting smooth value functions will converge to the classical solution of the original system of DPEs. The construction and approximation of the optimal feedback strategy for the original problem are also thoroughly discussed.

\vspace{0.3 cm}


\vspace{0.3 cm}


\end{abstract}

\ \\
\section{Introduction}

One ultimate goal for the community of financial mathematics is to characterize the sophisticated investment environment using tractable probabilistic or stochastic models. For example, the market trend is usually described by some random factors such as Markov chains. %The so-called regime-switching formulation is widely accepted in a lot of scientific research and provides some successful explanations consistent with many empirical observations.
In particular, the so-called regime-switching model is widely accepted and usually proposed to capture the influence on the behavior of the market caused by transitions in the macroeconomic system or the macroscopic readjustment and regulation. For instance, the empirical results by Ang and Bekaert~\cite{AngBeK02b} illustrate the existence of two regimes characterized by different levels of volatility. It is well known that default events modulated by the regime-switching process have an impact on the distress state of the surviving securities in the portfolio. More specifically, by an empirical study of the corporate bond market over 150 years, Giesecke et al.~\cite{GieSchStr11} suggest the existence of three regimes corresponding to high, middle, and low default risk. With finitely many economical regimes, Capponi and Figueroa-L\'opez~\cite{CapLop14a} investigate the classical utility maximization problem from terminal wealth based on a defaultable security, and Capponi, Figueroa-L\'opez and Nisen~\cite{CapLopNis14b} obtain a Poisson series representation for the arbitrage-free price process of vulnerable contingent claims.

On the other hand, the importance of considering the defaultable underlying assets has attracted a lot of attention, especially after the systemic failure caused by some global financial crisis. Some recent developments extend the early model of single defaultable security to default contagion effects on portfolio allocations. The research of these mutual contagion influence opens the door to provide possible answers to some empirical puzzles like the high mark-to-market variations in prices of credit sensitive assets. For example, Kraft and Steffensen~\cite{Kraf} discuss the contagion effects on defaultable bonds. Callegaro, Jeanblanc and Runggaldier~\cite{Call12} consider an optimal investment problem with multiple defaultable assets which depend on a partially observed exogenous factor process. Jiao, Kharroubi and Pham~\cite{Jiao13} study the model in which multiple jumps and default events are allowed. Recently, Bo and Capponi~\cite{Bo16} examine the optimal portfolio problem of a power utility investor who allocates the wealth between credit default swaps and a money market for which the contagion risk is modeled via interacting default intensities.

Apart from the celebrated Merton's model on utility maximization, there has been an increasing interest in the risk-sensitive stochastic control criterion in the portfolio management during recent years, see, e.g., Davis and Lleo~\cite{DavisLIeo04} for an overview of the theory and practice of risk-sensitive asset management. In a typical risk sensitive portfolio optimization problem, the investor maximizes the long run growth rate of the portfolio, adjusted by a measure of volatility. In particular, the classical utility maximization from terminal wealth can be transformed to the risk-sensitive control criterion by introducing a change of measure and a so-called risk-sensitive parameter which characterizes on the degree of risk tolerance of investors, see, e.g., Bielecki and Pliska~\cite{BiePliska99} and Nagai and Peng~\cite{PengNagai}. We will only name a small portion of the vast literature, for instance, the risk sensitive criterion can be linked to the dynamic version of Markowitz's mean-variance optimization by Bielecki and Pliska~\cite{BiePliska99}, to differential games by Fleming~\cite{Fleming06} and more recently by Bayraktar and Yao~\cite{BayraktarYao13} for the connection to zero-sum stochastic differential games using BSDEs and the weak dynamic programming principle. Hansen, et al.~\cite{HansenNoa} further connect the risk-sensitive objective to a robust criteria in which perturbations are characterized by the relative entropy. Bayraktar and Cohen~\cite{BayraktarCohen16} later examine a risk sensitive control version of the lifetime ruin probability problem.

Despite many existing work on the risk-sensitive control, optimal investment with credit risk or regime switching respectively, it remains an open problem of the risk-sensitive portfolio allocation with both scenarios of default risk and regime-switching. Our paper aims to fill this gap and considers an interesting case when the default contagion effect can depend on regime states, possibly infinitely many. For some recent related work, it is worth noting that in the default-free market with finite regime states, Andruszkiewicz, Davis and Lleo~\cite{AndDavLIeo} study the existence and uniqueness of the solution to the risk-sensitive asset maximization problem, and provide an ODE for the optimal value function, which may be efficiently solved numerically. Meanwhile, Das, Goswami and Rana~\cite{DasGosRan} consider a risk-sensitive portfolio optimization problem with multiple stocks modeled as a multi-dimensional jump diffusion whose coefficients are modulated by an age-dependent semi-Markov process. They also establish the existence and uniqueness of classical solutions to the corresponding HJB equations. In the context of theoretical stochastic control, we also note that Kumar and Pal~\cite{KumarPal} derive the dynamical programming principle for a class of risk-sensitive control problem of pure jump process with near monotone cost. To model hybrid diffusions, Nguyen and Yin~\cite{NguyenYin} propose a switching diffusion system with countably infinite states. The existence and uniqueness of the solution to the hybrid diffusion with past-dependent switching are obtained. Back to the practical implementation in financial markets with stochastic factors, the regime-switching model or continuous time Markov chain is frequently used to approximate the dynamics of time-dependent market parameter or factors. The continuous state space of the parameter or factor is usually discretized which lead to infinite states of the approximating Markov chain (see, e.g., Ang and Timmermann~\cite{AngTim}). This mainly motivates us to consider the countable regime states in this work and it is shown that this technical difficulties can eventually be reconciled using an appropriate approximation by counterparts with finite states. Therefore, our analytical conclusions for regime-switching can potentially provide theoretical foundations for numerical treatment of risk sensitive portfolio optimization with defaults and stochastic factor processes.

Our contributions are twofold. From the modeling perspective, it is considered that the correlated stocks are subject to credit events, and in particular, the dynamics of defaultable stocks, namely the drift, the volatility and the default intensity coefficients, all depend on the macroeconomic regimes. As defaults can occur sequentially, the default contagion is modeled in the sense that default intensities of surviving names are affected simultaneously by default events of other stocks as well as on current regimes states. This set up in our model enables us to analyze the joint complexity rooted in the investor's risk sensitivity, the regime changes and the default contagion among stocks. From the mathematical perspective, the resulting dynamic programming equation (DPE) can be viewed as a recursive infinite-dimensional nonlinear dynamical system in terms of default states. The depth of the recursion equals the number of stocks in the portfolio. Our recipe to study this new type of recursive dynamical system can be summarized in the following scheme: First, it is proposed to truncate the countably infinite state space of the regime switching process and consider the recursive DPE only with a finite state space. Second, for the finite state case, the existence and uniqueness of the solutions of the recursive DPE are analyzed based upon a backward recursion, namely from the state in which all stocks are defaulted toward the state in which all stocks are alive. It is worth noting that no bounded constraint is reinforced on the trading strategies of securities or control variables as in Andruszkiewicz, Davis and Lleo~\cite{AndDavLIeo} and Kumar and Pal~\cite{KumarPal}. As a price to pay, the nonlinearities of the HJB dynamical systems are not globally Lipschitz continuous. To overcome this new challenge, we develop a truncation technique by proving a comparison theorem based on the theory of monotone dynamical systems documented in Smith~\cite{smith08}. Then, we establish a unique classical solution of the recursive DPE by showing that the solution of truncated system has a uniform (strictly positive) lower bound independent of the truncation level. This also enables us to characterize the optimal admissible feedback trading strategy in the verification theorem. Next, when the states are relaxed to be countably infinite, the results in the finite state case can be applied to construct a sequence of approximating risk sensitive control problems to the original problem and obtain elegant uniform estimates to conclude that the sequence of associated smooth value functions will successfully converge to the classical solution of the original recursive DPE. We also contribute to the existing literature by exploring the possible construction and approximation of the optimal feedback strategy in some rigorous verification theorems.

%To compensate the theoretical discoveries, we provide some numerical analysis to confirm the analytical convergence results. A simple model with two defaultable stocks and a riskless bond is considered to pin down the major contagion effects and regime switching impacts. By applying the finite-difference scheme to a sequence of recursive DPEs with increasing finite states, the monotonicity of the approximating value functions and the convergence of both value functions and optimal strategies to the original problem are checked graphically. Moreover, the sensitivity analysis of the optimal portfolio choices are also performed on the risk sensitive parameters, the volatility of stocks and the default contagion effects. Our numerical observations are consistent with the common intuitions and conjectures.

The rest of the paper is organized as follows. Section \ref{sec:model} describes the credit market model with default contagion and regime switching. Section \ref{risksens} formulates the risk-sensitive stochastic control problem and introduces the corresponding DPE. We analyze the existence and uniqueness of the classical global solution of recursive infinite-dimensional DPEs and develop rigorous verification theorems in Section \ref{sec:mainres}. %To confirm the theoretical conclusions, numerical implementations and sensitivity analysis of optimal portfolio strategies are performed in Section \ref{sec:numerics}.
For the completeness, some auxiliary results and proofs are delegated to the Appendix~\ref{app:proof1}.

\section{The Model} \label{sec:model}

We consider a model of the financial market consisting of $N\geq1$ defaultable stocks and a risk-free money market account on a given complete filtered probability space $(\Omega,{\mathcal G},{\Gx},\Px)$. Let $Y=(Y(t))_{t\in[0,T]}$ be a regime-switching process which will be introduced precisely later. The global filtration $\Gx=\Fx\vee{\Hx}$ augmented by all $\Px$-null sets satisfies the usual conditions. The filtration $\Fx=({\mathcal{F}}_t)_{t\in[0,T]}$ is jointly generated by the regime-switching process $Y$ and an independent $d\geq1$-dimensional Brownian motions denoted by
$W=(W_j(t);\ j=1,\ldots,d)_{t\in[0,T]}^{\top}$. We use $\top$ to denote the transpose operator. The time horizon of the investment is given by $T>0$.

The price process of the money market account $B(t)$ satisfies $dB(t)= r(Y(t))B(t)dt$, where $r(Y(t))\geq0$ is interest rate modulated by the regime-switching process $Y$. The filtration $\Hx$ is generated by a $N$-dimensional default indicator process $Z=(Z_j(t);\ j=1,\ldots,N)_{t\in[0,T]}$ which
takes values in ${\cal S}:=\{0,1\}^N$. The default indicator process $Z$ links to the default times of the $N$ defaultable stocks via $\tau_j := \inf\{t\geq0;\ Z_j(t)=1\}$
for $j=1,\ldots,N$. The filtration $\Hx=({\mathcal{H}}_t)_{t\in[0,T]}$ is defined by ${\cal H}_t=\bigvee_{j=1}^N{\sigma(Z_j(s);\ s\leq t)}$.
Hence $\Hx$ contains all information about default events until the terminal time $T$. The market model is specified in detail in the following subsections.

\subsection{Regime-Switching Process}\label{sub:RSP} The regime-switching process is described by a continuous time (conservative) Markov chain $Y=(Y(t))_{t\in[0,T]}$ with countable state space
$\Zx_+:=\N\setminus\{0\}=\{1,2,\ldots\}$. The generator of the Markov chain $Y$ is given by the $Q$-matrix $Q=(q_{ij})_{ij\in\Zx_+}$. This yields that $q_{ii}\leq0$ for $i\in\Zx_+$, $q_{ij}\geq0$ for $i\neq j$, and $\sum_{j=1}^{\infty}q_{ij}=0$ for $i\in\Zx_+$ (i.e., $\sum_{j\neq i}q_{ij}=-q_{ii}$ for $i\in\Zx_+$).

\subsection{Credit Risk Model} The joint process $(Y,Z)$ of the regime-switching process and the default indicator process is a Markov process on the state space $\Zx_+\times\mathcal{S}$.
Moreover, at time $t$, the default indicator process transits from a state $Z(t):=(Z_1(t),\ldots,Z_{j-1}(t),Z_j(t),Z_{j+1}(t),\ldots,Z_N(t))$
%\begin{equation}
%Z(t):=(Z_1(t),\ldots,Z_{j-1}(t),Z_j(t),Z_{j+1}(t),\ldots,Z_N(t))\nonumber
%\end{equation}
in which the obligor $j$ is alive ($Z_j(t)=0$) to the neighbor state ${Z}^j(t):=(Z_1(t),\ldots,Z_{j-1}(t),1-Z_j(t),Z_{j+1}(t),\ldots,Z_N(t))$ in which the obligor $j$ has defaulted at a strictly positive stochastic rate $\lambda_{j}(Y(t),Z(t))$. We assume that $Y$ and $Z_1,\ldots,Z_N$ will not jump simultaneously. Therefore, the default intensity of the $j$-th stock may change either if any other stock in the portfolio defaults (contagion effect), or if there are regime-switchings. Our default model thus belongs to the rich class of interacting intensity models, introduced by Frey and Backhaus~\cite{FreyBackhaus04}. We set $\lambda(i,z)=(\lambda_j(i,z);\ j=1,\ldots,N)^{\top}$ for $(i,z)\in\Zx_+\times{\cal S}$.

\subsection{Price Processes} The price process of the $N$ defaultable stocks is denoted by the vector process $\tilde{P}=(\tilde{P}^j(t);\ j=1,\ldots,N)_{t\in[0,T]}^{\top}$. Here the price process of the $j$-th stock is given by, for $t\in[0,T]$,
\begin{equation}\label{eq:pricedef}
\tilde{P}_j(t)=(1-Z_j(t))P_j(t), \ \ \ j = 1,\ldots,N,
\end{equation}
where $P=(P_j(t);\ j=1,\ldots,N)_{t\in[0,T]}^{\top}$ represents the pre-default price of the $N$ stocks. In particular, the price of the $j$-th stock is given by the pre-default price $P_j(t)$ up to ${\tau_j}-$, and jumps to $0$ at default time ${\tau_j}$ and remains at $0$ afterwards. The pre-default price process $P$ of the $N$ defaultable stocks is assumed to satisfy
\begin{align}\label{eq:P}
dP(t) = {\rm diag}(P(t)) [(\mu(Y(t))+\lambda(Y(t),Z(t))) dt + \sigma(Y(t))dW(t)],
\end{align}
where, ${\rm diag}(P(t))$ is the diagonal $N\times N$-dimensional matrix with diagonal elements $P_i(t)$. For each $i\in\Zx_+$, the vector $\mu(i)$ is $\R^N$-valued column vector and $\sigma(i)$ is $\R^{N\times d}$-valued matrices such that $\sigma(i)\sigma(i)^\top$ is positive definite.
By Eq.s~\eqref{eq:pricedef}, \eqref{eq:P} and integration by parts, the price dynamics of defaultable stocks satisfies that
\begin{align}\label{eq:tildeP}
d\tilde{P}(t) = {\rm diag}(\tilde{P}(t)) [\mu(Y(t))dt +  \sigma(Y(t))dW(t)-dM(t)].
\end{align}
Here, $M=(M_j(t);\ j=1,\ldots,N)_{t\in[0,T]}^{\top}$ is a pure jump $\Gx=(\G_t)_{t\in[0,T]}$-martingale given by
\begin{align}\label{eq:taui}
M_j(t)&:= Z_j(t) - \int_0^{t\wedge\tau_j}\lambda_j(Y(s),Z(s))ds,\ \ \ \ \ \ t\in[0,T].
\end{align}
By the construction of the default indicator process $Z$ in Bo and Capponi~\cite{BoCapponi18}, it can be seen that $W$ is also a $\Gx$-Brownian motion using the
condition (M.2a) in Section 6.1.1 of Chapter 6 in Bielecki and Rutkowski~\cite{BieRut04}.

\section{Dynamic Optimization Problem} \label{risksens}

In this section, we formally derive the dynamic programming equation (DPE) associated with the risk sensitive stochastic control problem. We first reformulate the risk sensitive portfolio optimization problem in an equivalent form in Section \ref{sec:wealth}. The corresponding DPE will be derived and analyzed in Section \ref{sec:DPE}.

\subsection{Formulation of Portfolio Optimization Problem} \label{sec:wealth}

Let us first introduce the set up and formulate the risk sensitive portfolio optimization problem. For $t\in[0,T]$, let $\phi_B(t)$ represent the number of shares of the risk-free asset and let $\phi_j(t)$ denote the number of shares of the $j$-th stock at time $t$ held by the investor. The resulting wealth process is given by
\begin{align*}
X^{\phi}(t) = \sum_{j=1}^N\phi_j(t)\tilde{P}_j(t) + \phi_B(t)B(t),\ \ \ t\in[0,T].
\end{align*}
Using the price representation~\eqref{eq:pricedef} of stocks, the above wealth process can be rewritten as:
\begin{align*}%\label{eq:wealth0}
X^{\phi}(t)=\sum_{j=1}^N\phi_j(t) {(1-Z_j(t))} P_j(t)+\phi_B(t)B(t),\ \ \ t\in[0,T].
\end{align*}
For a given positive wealth process, we can consider the fractions of wealth invested in the stocks and money market account as follows: for $j=1,\ldots,N$, let us define $\tilde{\pi}_j(t)=\frac{\phi_j(t)\tilde{P}_j(t-)}{X^{\phi}(t-)}$ and $\tilde{\pi}_B(t)=1-\tilde{\pi}(t)^{\top}e_N$, where $\tilde{\pi}(t)=(\tilde{\pi}_i(t);\ i=1,\ldots,N)^{\top}$,
and $e_N = \big(\underbrace{1,1,\ldots,1}_{N \; ones}\big)^{\top}$.
Noting that the price of the $j$-th stock jumps to zero when the $j$-th stock defaults, the fraction of wealth held by the investor in this stock is zero after it defaults.
In particular, the following equality holds {$\tilde{\pi}_j(t)=(1-Z_j(t-))\tilde{\pi}_j(t)$ for $j=1,\ldots,N$}. Therefore, the self-financing condition leads to wealth dynamics in the following form: $X^{\tilde{\pi}}(0)=x\in\R_+:=(0,\infty)$, and
\begin{align}\label{eq:wealth}
dX^{\tilde{\pi}}(t) &= X^{\tilde{\pi}}(t-)\tilde{\pi}(t)^{\top}{\rm diag}(\tilde{P}(t-))^{-1}d\tilde{P}(t) + X^{\tilde{\pi}}(t)(1-\tilde{\pi}(t)^{\top}e_N)\frac{dB(t)}{B(t)}\\
&=X^{\tilde{\pi}}(t)\big[r(Y(t))+\tilde{\pi}(t)^{\top}(\mu(Y(t))-r(Y(t))e_N)\big]dt+ X^{\tilde{\pi}}(t-)\tilde{\pi}(t)^{\top}[\sigma(Y(t))dW(t)-dM(t)].\nonumber
\end{align}
%where $\bar{\mu}(y,z):=\mu(y)-re_n$ for $(y,z)\in\R^m\times{\cal S}$.

We next introduce the definition of the set of all admissible controls used in the paper.
\begin{definition}\label{def:add-con}
The admissible control set $\tilde{\cal U}$ is a class of $\Gx$-predictable feedback strategies $\tilde{\pi}(t)=(\tilde{\pi}_j(t);\ j=1,\ldots,N)^{\top}$, $t\in[0,T]$,
given by $\tilde{\pi}_j(t)=\pi_j(t,X^{\tilde{\pi}}(t-),Y(t-),Z(t-))$ such that SDE~\eqref{eq:wealth} admits a unique positive (strong) solution for $X^{\tilde{\pi}}(0)=x\in\R_+$ (i.e. the feedback strategies $\tilde{\pi}(t)$ should take values in $U:=(-\infty,1)^N$). Furthermore, the control $\tilde{\pi}=(\tilde{\pi}(t))_{t\in[0,T]}$ is required to make the positive process $\Gamma^{\tilde{\pi},\theta}=(\Gamma^{\tilde{\pi},\theta}(t))_{t\in[0,T]}$
defined later by \eqref{eq:Gam} to be a $\Px$-martingale.
\end{definition}
We will prove the martingale property of $\Gamma^{\tilde{\pi}^*,\theta}$ for a candidate optimal strategy $\tilde{\pi}^*$ by verifying the generalized Novikov's condition in Section~\ref{sec:mainres}. We consider the following {risk-sensitive} objective functional. For $\tilde{\pi}\in\tilde{\cal U}$, and given the initial values $(X(0),Y(0),Z(0))=(x,i,z)\in\R_+\times\Zx_+\times{\cal S}$, we define
\begin{align}\label{eq:J0}
{\cal J}(\tilde{\pi};T,x,i,z) := -\frac{2}{\theta}\log\Ex\left[\exp\left(-\frac{\theta}{2}\log X^{\tilde{\pi}}(T)\right)\right]=-\frac{2}{\theta}\log\Ex\left[(X^{\tilde{\pi}}(T))^{-\frac{\theta}{2}}\right].
\end{align}
The investor aims to maximize the objective functional ${\cal J}$ over all strategies $\tilde{\pi}\in\tilde{\cal U}$. Let us only focus on the case when $\theta\in(0,\infty)$ for a risk-sensitive investor.
The case $\theta\in(-2,0)$ is ignored as it is associated to a risk-seeking behavior which is less encountered in practise.
The objective functional \eqref{eq:J0} has been considered in the existing literature (see, e.g., Bielecki and Pliska~\cite{BiePliska99})
for dynamic asset allocations in the presence of market risk, however, it is still an open problem in the setting with default risk and regime-switching which motivates our research of this project. Eq.~(1.1) in Bielecki and Pliska~\cite{BiePliska99} in our case can be read as: for $\theta$ close to $0$,
\begin{align}\label{eq:rem-000}
{\cal J}(\tilde{\pi};T,x,y,z)=\Ex\left[\log\left(X^{\tilde{\pi}}(T)\right)\right]-\frac{\theta}{4}{\rm Var}\left(\log(X^{\tilde{\pi}}(T))\right)+o(\theta^2),
\end{align}
where $o(\theta^2)$ will typically depend on the terminal horizon $T$. % and $U(v)=\log v$ with $v>0$.
Then ${\cal J}(\tilde{\pi};T,x,y,z)$ may be interpreted as the growth rate of the investor's wealth minus a penalty term proportional to the variance of the realized rate, with an error that is proportional to $\theta^2$. This establishes a link between the risk-sensitive control problem and the robust decision making rule. A risk-sensitive investor would like to design a decision rule which protects him against large deviations of the growth rate from its expectation, and he achieves this by choosing higher values of the parameter $\theta$.

We next rewrite the objective functional as the exponential of an integral criterion (similar to Nagai and Peng~\cite{PengNagai}, and Capponi et al.~\cite{CappPascucci}) which will turn out to be convenient for the analysis of the dynamic programming equation. For all $\tilde{\pi}\in\tilde{\cal U}$, the wealth process solving SDE~\eqref{eq:wealth} is given by
\begin{align*}
X^{\tilde{\pi}}(T)=&x\exp\Bigg\{\int_0^T\big[r(Y(s))+\tilde{\pi}^{\top}(s)({\mu}(Y(s))-r(Y(s))e_N)\big]ds+\int_0^T\tilde{\pi}^{\top}(s)\sigma(Y(s))dW(s)\nonumber\\
&-\frac{1}{2}\int_0^T\left\|\sigma(Y(s))^{\top}\tilde{\pi}(s)\right\|^2ds+\sum_{j=1}^N\int_0^T\log(1-\tilde{\pi}_j(s))dM_j(s)\nonumber\\
&+\sum_{j=1}^N\int_0^{T\wedge\tau_j}\lambda_j(Y(s),Z(s))\big[\tilde{\pi}_j(s)+\log(1-\tilde{\pi}_j(s))\big]ds\Bigg\},
\end{align*}
and consequently
\begin{align}\label{eq:solution}
\left(X^{\tilde{\pi}}(T)\right)^{-\frac{\theta}{2}}
&=x^{-\frac{\theta}{2}}\Gam^{\tilde{\pi},\theta}(T)\exp\left(\frac{\theta}{2}\int_0^TL(\tilde{\pi}(s);Y(s),Z(s))ds\right),
\end{align}
where, for $(\pi,i,z)\in U\times\Zx_+\times{\cal S}$, the risk sensitive function $L(\pi;i,z)$ is defined by
\begin{align}\label{eq:L0}
L(\pi;i,z)&:= -r(i)-\pi^{\top}(\mu(i)-r(i)e_N)+\frac{1}{2}\left(1+\frac{\theta}{2}\right)\left\|\sigma(i)^{\top}\pi\right\|^2\nonumber\\
&\quad-\sum_{j=1}^N(1-z_j)\left[\frac{2}{\theta}+\pi_j-\frac{2}{\theta}(1-\pi_j)^{-\frac{\theta}{2}}\right]\lambda_j(i,z).
\end{align}
Here, the positive density process is given by, for $t\in[0,T]$,
\begin{align}\label{eq:Gam}
\Gam^{\tilde{\pi},\theta}(t)&:={\cal E}(\Pi^{\tilde{\pi},\theta})_t,\\
\Pi^{\tilde{\pi},\theta}(t)&:=-\frac{\theta}{2}\int_0^t\tilde{\pi}(s)^{\top}\sigma(Y(s))dW(s)+\sum_{j=1}^N\int_0^t\{(1-\tilde{\pi}_j(s))^{-\frac{\theta}{2}}-1\}dM_j(s),\nonumber
\end{align}
where ${\cal E}(\cdot)$ denotes the stochastic exponential.

As $\tilde{\pi}\in\tilde{\cal U}$, we have that $\Gam^{\tilde{\pi},\theta}=(\Gam^{\tilde{\pi},\theta}(t))_{t\in[0,T]}$ is a $\Px$-martingale.
We can thus define the following change of measure given by
\begin{align}\label{eq:change-measure}
\frac{d\Px^{\tilde{\pi},\theta}}{d\Px}\big|_{\G_t}=\Gam^{\tilde{\pi},\theta}(t),\ \ \ \ \ \ t\in[0,T],
\end{align}
under which
\begin{align}\label{eq:BMtheta}
W^{\tilde{\pi},\theta}(t):=W(t)+\frac{\theta}{2}\int_0^t\sigma(Y(s))^{\top}\tilde{\pi}(s)ds,\ \ \ \ \ \ t\in[0,T]
\end{align}
is a $d$-dimensional Brownian motion, while under $\Px^{\tilde{\pi},\theta}$, for $j=1,\ldots,N$, it holds that
\begin{align}\label{eq:Girjump}
M_j^{\tilde{\pi},\theta}(t):=Z_j(t)-\int_0^{t\wedge\tau_j}(1-\tilde{\pi}_j(s))^{-\frac{\theta}{2}}\lambda_j(Y(s),Z(s))ds,\qquad t\in[0,T]
\end{align}
is a martingale. The definition of $\Px^{\tilde{\pi},\theta}$ enables us to rewrite the above {risk-sensitive} objective functional~\eqref{eq:J0} in an exponential form. From~\eqref{eq:solution}, we deduce that
\begin{align*}\label{eq:J2}
{\cal J}(\tilde{\pi};T,x,i,z) &= -\frac{2}{\theta}\log\Ex\left[\left(X^{\tilde{\pi}}(T)\right)^{-\frac{\theta}{2}}\right]=-\frac{2}{\theta}\log\Ex\left[x^{-\frac{\theta}{2}}\Gam^{\tilde{\pi},\theta}(T)\exp\left(\frac{\theta}{2}\int_0^TL(\tilde{\pi}(s);Y(s),Z(s))ds\right)\right]\nonumber\\
&=\log x -\frac{2}{\theta}\log\Ex^{\tilde{\pi},\theta}\left[\exp\left(\frac{\theta}{2}\int_0^TL(\tilde{\pi}(s);Y(s),Z(s))ds\right)\right]=:\log x + \bar{{\cal J}}(\tilde{\pi};T,i,z).
\end{align*}
Here $\Ex^{\tilde{\pi},\theta}$ represents the expectation w.r.t. $\Px^{\tilde{\pi},\theta}$ defined by \eqref{eq:change-measure}.
Thanks to the relationship between ${\cal J}$ and $\bar{{\cal J}}$, our original problem is equivalent to maximize $\bar{{\cal J}}$ over $\tilde{\pi}\in\tilde{\cal U}$. We can therefore reformulate the value function of the risk-sensitive control problem as:
\begin{equation}\label{eq:value-fcn}
V(T,i,z) = \sup_{\tilde{\pi}\in\tilde{\cal U}} \bar{{\cal J}}(\tilde{\pi};T,i,z)=-\frac{2}{\theta}\inf_{\tilde{\pi}\in\tilde{\cal U}} \log\Ex^{\tilde{\pi},\theta}\left[\exp\left(\frac{\theta}{2}\int_0^TL(\tilde{\pi}(s);Y(s),Z(s))ds\right)\right],
\end{equation}
for $(i,z)\in\Zx_+\times{\cal S}$.

\subsection{Dynamic Programming Equations} \label{sec:DPE}
In this section, we will first derive the dynamic programming equation (DPE) satisfied by the value function~\eqref{eq:value-fcn} using heuristic arguments in Birge et al.~\cite{BirBoCap17}. It will be postponed in the next section to show that the solution of DPE indeed coincides with the value function of our risk sensitive control problem in rigorous verification theorems.

Let $(t,i,z)\in[0,T]\times\Zx_+\times{\cal S}$ and define
\begin{equation}\label{eq:J}
\bar{V}(t,i,z) :=-\frac{2}{\theta}\inf_{\tilde{\pi}\in\tilde{\cal U}}\log J(\tilde{\pi};t,i,z):= -\frac{2}{\theta}\inf_{\tilde{\pi}\in\tilde{\cal U}}\log\Ex_{t,i,z}^{\tilde{\pi},\theta}\left[\exp\left(\frac{\theta}{2}\int_t^TL(\tilde{\pi}(s);Y(s),Z(s))ds\right)\right],
\end{equation}
where $\Ex_{t,i,z}^{\tilde{\pi},\theta}[\cdot]:=\Ex^{\tilde{\pi},\theta}[\cdot|Y(t)=i,Z(t)=z]$. This yields the relation ${V}(T,i,z)=\bar{V}(0,i,z)$.
For $0\leq t<s\leq T$, the dynamic programming principle leads to
\begin{equation}\label{eq:dpp}
\bar{V}(t,i,z)= -\frac{2}{\theta}\inf_{\tilde{\pi}\in\tilde{\cal U}}\log\Ex_{t,i,z}^{\tilde{\pi},\theta}\left[\exp\left(-\frac{\theta}{2}\bar{V}(s,Y(s),Z(s))+\frac{\theta}{2}\int_t^sL(\tilde{\pi}(u);Y(u),Z(u))du\right)\right].\nonumber
\end{equation}
{Using heuristic arguments in Birge et al.~\cite{BirBoCap17}, we have the following DPE satisfied by $\bar{V}$, i.e., for all $(t,i,z)\in[0,T)\times\Zx_+\times{\cal S}$,
\begin{align}\label{eq:dpe2}
0=&\frac{\partial \bar{V}(t,i,z)}{\partial t}-\frac{2}{\theta}\sum_{l\neq i}q_{il}\left[\exp\left(-\frac{\theta}{2}\big(\bar{V}(t,l,z)-\bar{V}(t,i,z)\big)\right)-1\right]\nonumber\\
&+\sup_{\pi\in{\cal U}}H\left(\pi;i,z,(\bar{V}(t,i,z^j);\ j=0,1,\ldots,N)\right)
\end{align}
with terminal condition $\bar{V}(T,i,z)=0$ for all $(i,z)\in\Zx_+\times{\cal S}$. In the above equation, the function $H$ is defined by, for $(\pi,i,z)\in U\times\Zx_+\times{\cal S}$,
\begin{align}\label{eq:H}
H(\pi;i,z,\bar{f}(z)):=&-\frac{2}{\theta}\sum_{j=1}^N\left[\exp\left(-\frac{\theta}{2}(f({z}^j)-f(z))\right)-1\right](1-z_j)(1-\pi_j)^{-\frac{\theta}{2}}\lambda_j(i,z)\nonumber\\
&+r(i)+\pi^{\top}(\mu(i)-r(i)e_N)-\frac{1}{2}\left(1+\frac{\theta}{2}\right)\left\|\sigma(i)^{\top}\pi\right\|^2\nonumber\\
&+\sum_{j=1}^N\left[\frac{2}{\theta}+\pi_j-\frac{2}{\theta}(1-\pi_j)^{-\frac{\theta}{2}}\right](1-z_j)\lambda_j(i,z).
\end{align}
Here $\bar{f}(z)=(f(z^j);\ j=0,1,\ldots,N)$ for any measurable function $f(z)$. Above, we used the notation ${z}^j:=(z_1,\ldots,z_{j-1},1-z_j,z_{j+1},\ldots,z_N)$ for $z\in{\cal S}$.}

Eq.~\eqref{eq:dpe2} is in fact a recursive system of DPEs. We consider the following Cole-Hopf transform of the solution given by
\begin{align}\label{eq:exp-trnas}
\varphi(t,i,z):=\exp\left(-\frac{\theta}{2}\bar{V}(t,i,z)\right),\qquad (t,i,z)\in[0,T]\times\Zx_+\times{\cal S}.
\end{align}
Then $\frac{\partial \varphi(t,i,z)}{\partial t}=-\frac{\theta}{2}\varphi(t,i,z)\frac{\partial \bar{V}(t,i,z)}{\partial t}$ for $(t,i,z)\in[0,T]\times\Zx_+\times{\cal S}$. Plugging it into Eq.~\eqref{eq:dpe2}, we get that
\begin{align}\label{eq:dpe3}
0=&\frac{\partial \varphi(t,i,z)}{\partial t}+\sum_{l\neq i}q_{il}\left[\varphi(t,l,z)-\varphi(t,i,z)\right]+\inf_{\pi\in U}\tilde{H}\left(\pi;i,z,(\varphi(t,i,z^j);\ j=0,1,\ldots,N)\right)
\end{align}
with terminal condition $\varphi(T,i,z)=1$ for all $(i,z)\in\Zx_+\times{\cal S}$. In the above equation, the function $\tilde{H}$ is defined by
\begin{align}\label{eq:H}
\tilde{H}(\pi;i,z,\bar{f}(z))%=&\sum_{j=1}^N\left[f(z^j)-f(z)\right](1-z_j)(1-\pi_j)^{-\frac{\theta}{2}}\lambda_j(i,z)\nonumber\\
%&+\Bigg\{-\frac{\theta}{2}r(i)-\frac{\theta}{2}\pi^{\top}(\mu(i)-r(i)e_N)+\frac{\theta}{4}\left(1+\frac{\theta}{2}\right)\left\|\sigma(i)^{\top}\pi\right\|^2\nonumber\\
%&+\sum_{j=1}^N\left[(1-\pi_j)^{-\frac{\theta}{2}}-1-\frac{\theta}{2}\pi_j\right](1-z_j)\lambda_j(i,z)\Bigg\}f(z)\nonumber\\
:=&\Bigg\{-\frac{\theta}{2}r(i)-\frac{\theta}{2}\pi^{\top}(\mu(i)-r(i)e_N)+\frac{\theta}{4}\left(1+\frac{\theta}{2}\right)\left\|\sigma(i)^{\top}\pi\right\|^2
\\
&+\sum_{j=1}^N\left(-1-\frac{\theta}{2}\pi_j\right)(1-z_j)\lambda_j(i,z)\Bigg\}f(z)+\sum_{j=1}^Nf(z^j)(1-z_j)(1-\pi_j)^{-\frac{\theta}{2}}\lambda_j(i,z).\nonumber
\end{align}

\section{Main Results and Verification Theorems}\label{sec:mainres}

We analyze the existence of global solutions of the recursive system of DPEs \eqref{eq:dpe3} in a two-step procedure. Firstly, we investigate the existence and uniqueness of classical solutions of Eq.~\eqref{eq:dpe3} as a dynamical system when the Markov chain $Y$ takes values in the finite state space. Secondly, we proceed to study the countably infinite state case using approximation arguments.

Let us introduce some notations which will be used frequently in this section. Let $n\in\Zx_+$. For $x\in\mathbb{R}^n$, we write $x=(x_1,...,x_n)^{\top}$. For any $x,y\in\R^n$, we write $x\leq y$ if $x_i\leq y_i$ for all $i=1,\ldots,n$, while write $x<y$ if $x\leq y$ and there exists some $i\in\{1,\ldots,n\}$ such that $x_i<y_i$. In particular, $x\ll y$ if $x_i<y_i$ for all $i=1,...,n$. Recall that $e_{N}$ denotes the $N$-dimensional column vector whose all entries are ones. For the general default state $z\in{\cal S}$, we here introduce a general default state representation $z=0^{j_1,\ldots,j_k}$ for indices $j_1\neq\cdots\neq j_k$ belonging to $\{1,\ldots,N\}$, and $k\in\{0,1,\ldots,N\}$. Such a vector $z$ is obtained by flipping the entries $j_1,\ldots,j_k$ of the zero vector to one, i.e. $z_{j_1}=\cdots=z_{j_k}=1$, and $z_{j}=0$ for $j\notin\{j_1,\ldots,j_k\}$ (if $k=0$, we set $z=0^{j_1,\ldots,j_k}=0$). Clearly $0^{j_1,\ldots,j_{N}}=e_N^{\top}$.

\subsection{Finite State Case of Regime-Switching Process}\label{sec:finite-states}
In this section, we study the case where the regime-switching process $Y$ is defined on a finite state space given by $D_n=\{1,\ldots,n\}$. Here $n\in\Zx_+$ is a fixed number. The corresponding $Q$-matrix of the Markov chain $Y$ is given by $Q_n=(q_{ij})_{i,j\in D_n}$ satisfying $\sum_{j\in D_n}q_{ij}=0$ for $i\in D_n$ and $q_{ij}\geq0$ when $i\neq j$. It is worth noting that $q_{ij}$, $i,j\in D_n$ here may be different from what is given in Subsection~\ref{sub:RSP}. With slight abuse of notation, we still use $q_{ij}$ here only for convenience.

Let $\varphi(t,z):=(\varphi(t,i,z);\ i=1,\ldots,n)^{\top}$ be a column vector of the solution for $(t,z)\in[0,T]\times{\cal S}$. Then, we can rewrite Eq.~\eqref{eq:dpe3} as the following dynamical system:
\begin{align}\label{eq:hjbeqn}
\left\{
\begin{aligned}
\frac{\partial \varphi(t,z)}{\partial t}+\big(Q_n+{\rm diag}(\nu(z))\big)\varphi(t,z)+G(t,\varphi(t,z),z)=&0,\quad t\in[0,T)\times{\cal S};\\
\varphi(T,z)=&e_n,\quad \text{for all }z\in{\cal S}.
\end{aligned}
\right.
\end{align}
Here, the vector of function $G(t,x,z)=(G_i(t,x,z);\ i=1,\ldots,n)^{\top}$ is given by, for each $i\in D_n$ and $(t,x,z)\in[0,T]\times\R^n\times{\cal S}$,
\begin{align}
G_i(t,x,z)=&\inf_{\pi\in U}\Bigg\{\sum_{j=1}^N\varphi(t,i,z^j)(1-z_j)(1-\pi_j)^{-\frac\theta2}\lambda_j(i,z)\\
&+\bigg[\frac\theta4(1+\frac\theta2)\left\|\sigma(i)^{\top}\pi\right\|^2-\frac\theta2\pi^\top(\mu(i)-r(i)e_N)-\frac{\theta}{2}\sum_{j=1}^N\pi_j(1-z_j)\lambda_j(i,z)\bigg]x_i\Bigg\}.\nonumber
\end{align}
The vector of coefficient $\nu(z)=(\nu_i(z);\ i=1,\ldots,n)^{\top}$ for $z\in{\cal S}$ is given by, for each $i\in D_n$,
\begin{align}\label{eq:nuz}
\nu_i(z)=-\frac{\theta}{2}r(i)-\sum_{j=1}^N(1-z_j)\lambda_j(i,z).
\end{align}

Recall the recursive system given by \eqref{eq:hjbeqn} in terms of default states $z=0^{j_1,\ldots,j_k}\in{\cal S}$ (where $k=0,1,\ldots,N$). The solvability can in fact be analyzed in the recursive form on default states. Therefore, our strategy for analyzing the system is based on a recursive procedure, starting from the default state $z=e_N^{\top}$ (i.e., all stocks have defaulted) and proceeding backward to the default state $z=0$ (i.e., all stocks are alive).
\begin{itemize}
  \item[(i)] $k=N$ (i.e., all stocks have defaulted). In this default state, it is clear that the investor will not invest in stocks and hence the optimal fraction strategy in stocks for this case is given by $\pi_1^*=\cdots=\pi_N^*=0$ by virtue of Definition~\ref{def:add-con}. Let $\varphi(t,e_N^{\top})=(\varphi(t,i,e_N^{\top});\ i=1,\ldots,n)^{\top}$. As a consequence, the dynamical system \eqref{eq:hjbeqn} can be written as
      \begin{align}\label{eq:hjben}
\left\{
\begin{aligned}
\frac d{dt}\varphi(t,e_N^{\top})=&-A^{(N)}\varphi(t,e_N^{\top}),\quad\text{ in }[0,T);\\
\varphi(T,e_N^{\top})=&e_n.
\end{aligned}
\right.
\end{align}
The matrix of coefficients $A^{(N)}:=Q_n+{\rm diag}(\nu(e_N^{\top}))$.
\end{itemize}
In order to establish the unique positive solution to the above dynamical system \eqref{eq:hjben}, we need the following auxiliary result.
\begin{lemma}\label{lem:sol-hjben2}
Let $g(t)=(g_i(t);\ i=1,\ldots,n)^{\top}$ satisfy the following dynamical system:
\begin{align*}
\left\{
\begin{aligned}
\frac d{dt}g(t)=&Bg(t)\quad\text{ in }(0,T];\\
g(0)=&\xi.
\end{aligned}
\right.
\end{align*}
If $B=(b_{ij})_{n\times n}$ satisfies $b_{ij}\geq 0$ for $i\neq j$ and $\xi\gg0$,
then we have $g(t)\gg0$ for all $t\in[0,T]$.
\end{lemma}

\noindent{\it Proof.}\quad Define $f(x)=Bx$ for $x\in\R^n$. By virtue of Proposition 1.1 of Chapter 3 in \cite{smith08}, it suffices to verify that $f:\R^n\to\R^n$ is of type $K$, i.e., for any $x,y\in\R^n$ satisfying $x\leq y$ and $x_i=y_i$ {for some $i=1,\ldots,n$}, then it holds that $f_i(x)\leq f_i(y)$. Notice that $b_{ij}\geq0$ for all $i\neq j$. Then, we have that
\begin{align}\label{eq:111}
f_i(x)&=(Bx)_i=\sum_{j=1}^nb_{ij}x_j=b_{ii}x_i+\sum_{j=1,j\neq i}^nb_{ij}x_j\nonumber\\
&=b_{ii}y_i+\sum_{j=1,j\neq i}^nb_{ij}x_j
\leq b_{ii}y_i+\sum_{j=1,j\neq i}^nb_{ij}y_j=f_i(y),
\end{align}
and hence $f$ is of type $K$. Thus, we complete the proof of the lemma. \hfill$\Box$\\

The next result is consequent on the previous lemma.
\begin{lemma}\label{lem:sol-hjben}
The dynamical system \eqref{eq:hjben} admits a unique solution which is given by
\begin{align}\label{eq:varphien}
\varphi(t,e_N^{\top})=  e^{A^{(N)}(T-t)}e_n=\sum_{i=0}^{\infty}\frac{(A^{(N)})^i(T-t)^i}{i!}e_n,\quad t\in[0,T],
\end{align}
where the $n\times n$-dimensional matrix $A^{(N)}= Q_n+{\rm diag}(\nu(e_N^{\top}))=Q_n-\frac{\theta}{2}{\rm diag}(r)$ with $r=(r(i);\ i=1,\ldots,n)^{\top}$. Moreover, it holds that $\varphi(t,e_N^{\top})\gg 0$ for all $t\in[0,T]$.
\end{lemma}

\noindent{\it Proof.}\quad The representation of the solution $\varphi(t,e_N^{\top})$ given by \eqref{eq:varphien} is obvious. Note that $e_n\gg0$ and $q_{ij}\geq0$ for all $i\neq j$ as $Q_n=(q_{ij})_{n\times n}$ is the generator of the Markov chain. Then in order to prove $\varphi(t,e_N^{\top})\gg0$ for all $t\in[0,T]$, using Lemma~\ref{lem:sol-hjben2}, it suffices to verify $[A^{(N)}]_{ij}\geq0$ for all $i\neq j$. However $[A^{(N)}]_{ij}=q_{ij}$ for all $i\neq j$ and the condition given in Lemma~\ref{lem:sol-hjben2} is therefore verified which implies that $\varphi(t,e_N^{\top})\gg0$ for all $t\in[0,T]$. \hfill$\Box$\\

We next consider the general default case with $z=0^{j_1,\ldots,j_{k}}$ for $0\leq k\leq N-1$, i.e. the stocks $j_1,\ldots,j_{k}$ have defaulted but the stocks $\{j_{k+1},\ldots,j_N\}:=\{1,\ldots,N\}\setminus\{j_1,\ldots,j_k\}$ remain alive. Then we have
\begin{itemize}
  \item[(ii)] Because the stocks $j_1,\ldots,j_k$ have defaulted, the optimal fraction strategies for the stocks $j_1,\ldots,j_{k}$ are given by $\pi_j^{(k,*)}=0$ for $j\in\{j_1,\ldots,j_{k}\}$ by virtue of Definition~\ref{def:add-con}. Let $\varphi^{(k)}(t)=(\varphi(t,i,0^{j_1,\ldots,j_{k}});\ i=1,\ldots,n)^{\top}$ and $\lambda^{(k)}_j(i)=\lambda_j(i,0^{j_1,\ldots,j_{k}})$ for $j\notin\{j_1,\ldots,j_k\}$ and $i=1,\ldots,n$. Then, the corresponding DPE \eqref{eq:hjbeqn} to this default case is given by
       \begin{align}\label{eq:hjbn-1}
\left\{
\begin{aligned}
\frac d{dt}\varphi^{(k)}(t)=&-A^{(k)}\varphi^{(k)}(t)-G^{(k)}(t,\varphi^{(k)}(t)),\quad\text{ in }[0,T);\\
\varphi^{(k)}(T)=&e_n.
\end{aligned}
\right.
\end{align}
Here, the $n\times n$-dimensional matrix $A^{(k)}$ is given by
\begin{align}\label{eq:An-1}
A^{(k)}={\rm diag}\left[\left(-\frac{\theta}{2} r(i)-\sum_{j\notin\{j_1,\ldots,j_{k}\}}\lambda_{j}^{(k)}(i);\ i=1,\ldots,n\right)\right]+Q_n.
\end{align}
The coefficient $G^{(k)}(t,x)=(G^{(k)}_i(t,x);\ i=1,\ldots,n)^{\top}$ for $(t,x)\in[0,T]\times\R^{n}$ is given by, for $i\in D_n$,
\begin{align}\label{eq:Gin-1}
G^{(k)}_i(t,x):=&\inf_{\pi^{(k)}\in U^{(k)}}\left\{\sum_{j\notin\{j_1,\ldots,j_k\}} \varphi^{(k+1),j}(t,i)\big(1-\pi_{j}^{(k)}\big)^{-\frac{\theta}{2}}\lambda_{j}^{(k)}(i)+H^{(k)}(\pi^{(k)};i)x_i\right\}.
\end{align}
where, for $(\pi^{(k)},i)\in U^{(k)}\times D_n$, the function $H^{(k)}$ is given by
\begin{align}\label{eq:Hk}
H^{(k)}(\pi^{(k)};i):=&\frac{\theta}{4}\big(1+\frac{\theta}{2}\big)\left\|\sigma^{(k)}(i)^{\top}\pi^{(k)}\right\|^2
-\frac{\theta}{2}(\pi^{(k)})^{\top}\big(\mu^{(k)}(i)-r(i)e_{N-k}\big)\nonumber\\
&-\frac{\theta}{2}\sum_{j\notin\{j_1,\ldots,j_k\}}\pi_{j}^{(k)}\lambda_{j}^{(k)}(i).
\end{align}
The policy space of this state is $U^{(k)}=(-\infty,1)^{N-k}$, and $\varphi^{(k+1),j}(t,i):=\varphi(t,i,0^{j_1,\ldots,j_k,j})$ for $j\notin\{j_1,\ldots,j_k\}$ corresponds to the $i$-th element of the positive solution vector of Eq.~\eqref{eq:hjbeqn} at the default state $z=0^{j_1,\ldots,j_k,j}$.
Here, for each $i=1,\ldots,n$, we have also used notations: $\pi^{(k)}=(\pi_j^{(k)};\ j\notin\{j_1,\ldots,j_k\})^{\top}$, $\theta^{(k)}(i)=(\theta_j(i);\ j\notin\{j_1,\ldots,j_k\})^{\top}$, $\sigma^{(k)}(i)=(\sigma_{j\kappa}(i);\ j\notin\{j_1,\ldots,j_k\},\kappa\in\{1,\ldots,d\})$, and $\mu^{(k)}(i)=(\mu_j(i);\ j\notin\{j_1,\ldots,j_k\})^{\top}$.
\end{itemize}

From the expression of $G_i^{(k)}(t,x)$ given by \eqref{eq:Gin-1}, it can be seen that the solution $\varphi^{(k)}(t)$ on $t\in[0,T]$ of DPE \eqref{eq:hjbeqn} at the default state $z=0^{j_1,\ldots,j_k}$ in fact depends on the solution $\varphi^{(k+1),j}(t)$ on $t\in[0,T]$ of DPE~\eqref{eq:hjbeqn} at the default state
$z=0^{j_1,\ldots,j_k,j}$ for $j\notin\{j_1,\ldots,j_k\}$. In particular when $k=N-1$, the solution $\varphi^{(k+1),j}(t)=\varphi(t,e_N^{\top})\gg0$ corresponds to the solution to \eqref{eq:hjbeqn} at the default state $z=e_N$ (i.e., $k=N$), which has been obtained by Lemma~\ref{lem:sol-hjben}.
This suggests us to solve DPE~\eqref{eq:hjbeqn} backward recursively in terms of default states $z=0^{j_1,\ldots,j_k}$. Thus, in order to study the existence and uniqueness of a positive (classical) solution to the dynamical system \eqref{eq:hjbn-1}, we first assume that \eqref{eq:hjbeqn} admits a positive unique (classical) solution $\varphi^{(k+1),j}(t)$ on $t\in[0,T]$ for $j\notin\{j_1,\ldots,j_k\}$.

We can first obtain an estimate on $G^{(k)}(t,x)$, which is presented in the following lemma.
\begin{lemma}\label{lem:Gkesti}
For each $k=0,1,\ldots,N-1$, let us assume that DPE~\eqref{eq:hjbeqn} admits a positive unique (classical) solution $\varphi^{(k+1),j}(t)$ on $t\in[0,T]$ for $j\notin\{j_1,\ldots,j_k\}$. Then, for any $x,y\in\R^n$ satisfying $x,y\geq\varepsilon e_n$ with $\varepsilon>0$, there exists a positive constant $C=C(\varepsilon)$ which only depends on $\varepsilon>0$ such that
\begin{align}\label{eq:Gkesti}
\left\|G^{(k)}(t,x)-G^{(k)}(t,y)\right\|\leq C\left\|x-y\right\|.
\end{align}
Here $\|\cdot\|$ denotes the Euclidian norm.
\end{lemma}

\noindent{\it Proof.}\quad  It suffices to prove that, for each $i=1,\ldots,n$, $|G^{(k)}_i(t,x)-G^{(k)}_i(t,y)|\leq C(\varepsilon)\|x-y\|$ for any $x,y\in\R^n$ satisfying $x,y\geq\varepsilon e_n$ with $\varepsilon>0$, where $C(\varepsilon)>0$ is independent of time $t$. By the recursive assumption, $\varphi^{(k+1),j}(t)$ on $t\in[0,T]$ is the unique positive (classical) solution to \eqref{eq:hjbeqn} for $j\notin\{j_1,\ldots,j_k\}$. Then, it is continuous on $[0,T]$ which implies the existence of a constant $C_0>0$ independent of $t$ such that $\sup_{t\in[0,T]}\|\varphi^{(k+1),j}(t)\|\leq C_0$ for $j\notin\{j_1,\ldots,j_k\}$. Thus, by \eqref{eq:Gin-1}, and thanks to $H^{(k)}(0;i)=0$ for all $i\in D_n$ using \eqref{eq:Hk}, it follows that, for all $(t,x)\in[0,T]\times\R^n$,
\begin{align}\label{eq:gless}
G^{(k)}_i(t,x)\leq&\left[\sum_{j\notin\{j_1,\ldots,j_k\}} \varphi^{(k+1),j}(t,i)(1-\pi_{j}^{(k)})^{-\frac{\theta}{2}}\lambda_{j}^{(k)}(i)+H^{(k)}(\pi^{(k)};i)x_i\right]\Bigg|_{\pi^{(k)}=0}\nonumber\\
=&\sum_{j\notin\{j_1,\ldots,j_k\}} \varphi^{(k+1),j}(t,i)\lambda_{j}^{(k)}(i)\leq C_0 \sum_{j\notin\{j_1,\ldots,j_k\}}\lambda_{j}^{(k)}(i).
\end{align}
On the other hand, as $\sigma^{(k)}(i)^\top\sigma^{(k)}(i)$ is positive-definite, there exists a positive constant $\delta>0$ such that $\big\|\sigma^{(k)}(i)^{\top}\pi^{(k)}\|^2\geq\delta\|\pi^{(k)}\|^2$ for all $i\in D_n$. Hence, the following estimate holds:
\begin{align}\label{eq:esti1}
&H^{(k)}(\pi^{(k)};i)\geq\frac{\theta}{4}(1+\frac{\theta}{2})\delta\left\|\pi^{(k)}\right\|^2-\frac{\theta}{2}\left(\left\|\mu^{(k)}(i)-r(i)e_{N-k}\right\|+\sum_{j\notin\{j_1,\ldots,j_k\}}
\lambda_{j}^{(k)}(i)\right)\left\|\pi^{(k)}\right\|.
\end{align}
We next take the positive constant defined as
\[
C_1:=2\frac{\left\|\mu^{(k)}(i)-r(i)e_{N-k}\right\|+\sum_{j\notin\{j_1,\ldots,j_k\}}\lambda_j^{(k)}(i)}{(1+\frac\theta2)\delta}.
\]
For all $\pi^{(k)}\in\{\pi^{(k)}\in U^{(k)};\ \|\pi^{(k)}\|\geq C_1\}$, it holds that
\begin{align}\label{eq:large0}
H^{(k)}(\pi^{(k)};i)\geq 0,\qquad i\in D_n.
\end{align}
This yields that, for all $\pi^{(k)}\in\{\pi^{(k)}\in U^{(k)};\ \|\pi^{(k)}\|\geq C_1\}$ and all $x\geq\varepsilon e_n$, we deduce from \eqref{eq:esti1} and \eqref{eq:large0} that
\begin{align*}
&\sum_{j\notin\{j_1,\ldots,j_k\}} \varphi^{(k+1),j}(t,i)(1-\pi_{j}^{(k)})^{-\frac{\theta}{2}}\lambda_{j}^{(k)}(i)+H^{(k)}(\pi^{(k)};i)x_i\geq H^{(k)}(\pi^{(k)};i)x_i\\
&\qquad\geq H^{(k)}(\pi^{(k)};i)\varepsilon\\
&\qquad\geq\varepsilon\left[\frac\theta4(1+\frac\theta2)\delta\left\|\pi^{(k)}\right\|^2-\frac\theta2\left(\left\|\mu^{(k)}(i)-r(i)e_{N-k}\right\|+\sum_{j\notin\{j_1,\ldots,j_k\}}\lambda_{j}^{(k)}(i)\right)
\left\|\pi^{(k)}\right\|\right].
\end{align*}
We shall choose another positive constant depending on $\varepsilon>0$ as
\[
C_2(\varepsilon):=\frac{C_1}2+\sqrt{\frac{C_1^2}4+\frac8{\varepsilon\theta(2+\theta)\delta}C_0\sum_{j\notin\{j_1,\ldots,j_k\}}\lambda_{j}^{(k)}(i)}.
\]
Then, for all $\pi^{(k)}\in\{\pi\in U^{(k)};\ \|\pi\|\geq C_2(\varepsilon)\}$ and all $x\geq\varepsilon e_n$, it holds that
\begin{align}\label{eq:esti002}
&\sum_{j\notin\{j_1,\ldots,j_k\}} \varphi^{(k+1),j}(t,i)(1-\pi_{j}^{(k)})^{-\frac{\theta}{2}}\lambda_{j}^{(k)}(i)+H^{(k)}(\pi^{(k)};i)x_i\geq C_0\sum_{j\notin\{j_1,\ldots,j_k\}}\lambda_{j}^{(k)}(i).
\end{align}
By \eqref{eq:gless}, we have that $G^{(k)}_i(t,x)\leq C_0\sum_{j\notin\{j_1,\ldots,j_k\}}\lambda_{j}^{(k)}(i)$ for all $(t,x)\in[0,T]\times\R^n$. Thus, it follows from \eqref{eq:esti002} that
\begin{align}\label{eq:G2}
G^{(k)}_i(t,x)&=\inf_{\pi^{(k)}\in U^{(k)}}\left\{\sum_{j\notin\{j_1,\ldots,j_k\}} \varphi^{(k+1),j}(t,i)(1-\pi_{j}^{(k)})^{-\frac{\theta}{2}}\lambda_{j}^{(k)}(i)+H^{(k)}(\pi^{(k)};i)x_i\right\}\\
&=\inf_{\substack{\pi^{(k)}\in\{\pi\in U^{(k)}:\\ \|\pi\|\leq C_2(\varepsilon)\}}}\left\{\sum_{j\notin\{j_1,\ldots,j_k\}} \varphi^{(k+1),j}(t,i)(1-\pi_{j}^{(k)})^{-\frac{\theta}{2}}\lambda_{j}^{(k)}(i)+H^{(k)}(\pi^{(k)};i)x_i\right\}.\nonumber
\end{align}
In virtue of \eqref{eq:G2}, it holds that
\begin{align}\label{eq:Gxy}
G^{(k)}_i(t,x)&=\inf_{\substack{\pi^{(k)}\in\{\pi\in U^{(k)}:\\ \|\pi\|\leq C_2(\varepsilon)\}}}\Bigg\{\sum_{j\notin\{j_1,\ldots,j_k\}} \varphi^{(k+1),j}(t,i)(1-\pi_{j}^{(k)})^{-\frac{\theta}{2}}\lambda_{j}^{(k)}(i)\nonumber\\
&\qquad\qquad\qquad\qquad+H^{(k)}(\pi^{(k)};i)y_i+H^{(k)}(\pi^{(k)};i)(x_i-y_i)\Bigg\}\nonumber\\
&\leq\inf_{\substack{\pi^{(k)}\in\{\pi\in U^{(k)}:\\ \|\pi\|\leq C_2(\varepsilon)\}}}\Bigg\{\sum_{j\notin\{j_1,\ldots,j_k\}} \varphi^{(k+1),j}(t,i)(1-\pi_{j}^{(k)})^{-\frac{\theta}{2}}\lambda_{j}^{(k)}(i)\nonumber\\
&\qquad\qquad\qquad\qquad+H^{(k)}(\pi^{(k)};i)y_i\Bigg\}+C(\varepsilon)|x_i-y_i|\nonumber\\
&= G^{(k)}_i(t,y)+C(\varepsilon)|x_i-y_i|.
\end{align}
Here, the finite positive constant $C(\varepsilon)=\max_{i=1,\ldots,n}C^{(i)}(\varepsilon)$, where for $i\in D_n$,
\begin{align}\label{eq:Cepsilon}
C^{(i)}(\varepsilon)&:=\sup_{\substack{\pi^{(k)}\in\{\pi\in U^{(k)}:\\ \|\pi\|\leq C_2(\varepsilon)\}}}H^{(k)}(\pi^{(k)};i).
\end{align}
Note that the constant $C^{(i)}(\varepsilon)$ given above is nonnegative and finite for each $i\in D_n$. By \eqref{eq:Gxy}, we get that
$|G^{(k)}_i(t,x)-G^{(k)}_i(t,y)|\leq C(\varepsilon)\|x-y\|$ for any $x,y\in\R^n$ satisfying $x,y\geq\varepsilon e_n$ with $\varepsilon>0$, which completes the proof of the lemma.
\hfill$\Box$\\

We move on to study the existence and uniqueness of the global (classical) solution to the dynamical system \eqref{eq:hjbn-1}. To this end, we prepare the following comparison results of two types of dynamical systems with the type $K$ condition introduced in Smith~\cite{smith08}:
\begin{lemma}\label{comparison}
Let $g_{\kappa}(t)=(g_{\kappa i}(t);\ i=1,\ldots,n)^{\top}$ with $\kappa=1,2$ satisfy the following dynamical systems on $[0,T]$, respectively
\begin{align*}
\left\{
\begin{aligned}
\frac d{dt}g_1(t)=&f(t,g_1(t))+\tilde{f}(t,g_1(t)),\ \text{ in }(0,T];\\
g_1(0)=&\xi_1,
\end{aligned}
\right.\qquad\qquad
\left\{
\begin{aligned}
\frac d{dt}g_2(t)=&f(t,g_2(t)),\ \text{ in }(0,T];\\
g_2(0)=&\xi_2.
\end{aligned}
\right.
\end{align*}
Here, the functions $f(t,x),\,\tilde{f}(t,x):[0,T]\times\R^n\to\R^n$ are assumed to be Lipschitz continuous w.r.t. $x\in\R^m$ uniformly in $t\in[0,T]$. The function
$f(t,\cdot)$ satisfies the type $K$ condition for each $t\in[0,T]$ (i.e., for any $x,y\in\R^n$ satisfying $x\leq y$ and $x_i=y_i$ for some $i=1,\ldots,n$,
it holds that $f_i(t,x)\leq f_i(t,y)$ for each $t\in[0,T]$). If
$\tilde{f}(t,x)\geq0$ for $(t,x)\in[0,T]\times\R^n$ and $\xi_1\geq\xi_2$, then $g_1(t)\geq g_2(t)$ for all $t\in[0,T]$.
\end{lemma}

\noindent{\it Proof.}\quad For $p>0$, let $g_{1}^{(p)}(t)=(g_{1i}^{(p)}(t);\ i=1,\ldots,n)^{\top}$ be the solution to the following dynamical system given by
\begin{equation}
\left\{
\begin{aligned}
\frac d{dt}g_{1}^{(p)}(t)=&f(t,g_{1}^{(p)}(t))+\tilde{f}(t,g^{(p)}_{1}(t))+\frac{1}{p}e_n^{\top},\ \text{ in }(0,T];\\
g_{1}^{(p)}(0)=&\xi_1+\frac{1}{p}e_n^{\top}.
\end{aligned}
\right.
\end{equation}
Then, for all $t\in(0,T]$, it holds that
\begin{align*}
\|g_{1}^{(p)}(t)-g_1(t)\|\leq&\|g_{1}^{(p)}(0)-g_1(0)\|+\int_0^t\big\|f(s,g_{1}^{(p)}(s))-f(s,g_1(s))\big\|ds\nonumber\\
&+\int_0^t\big\|\tilde{f}(s,g_{1}^{(p)}(s))-\tilde{f}(s,g_1(s))\big\|ds+\frac1p\int_0^t\|e_n\|ds\nonumber\\
\leq&\frac{1+T}p\|e_n\|+(C+\tilde{C})\int_0^t\big\|g_{1}^{(p)}(s)-g_1(s)\big\|ds.
\end{align*}
Here $C>0$ and $\tilde{C}>0$ are Lipschitz constant coefficients for $f(t,x)$ and $\tilde{f}(t,x)$, respectively. The Gronwall's lemma yields that $g_{1}^{(p)}(t)\to g_1(t)$ for all $t\in[0,T]$ as $p\to\infty$. We claim that $g_{1}^{(p)}(t)\gg g_2(t)$ for all $t\in[0,T]$. Suppose that the claim does not hold, the fact that $g_{1}^{(p)}(0)\gg g_2(0)$, and $g_1^{(p)}(t),g_2(t)$ are continuous on $[0,T]$ imply that there exists a $t_0\in(0,T]$ such that $g_{1}^{(p)}(s)\geq g_2(s)$ on $s\in[0,t_0]$ and $g_{1i}^{(p)}(t_0)=g_{2i}(t_0)$ for some $i\in\{1,\ldots,n\}$. Because for $t_0>0$, $g_1^{(p)}(t),g_2(t)$ are differentiable on $(0,T]$, it follows that
\begin{align*}
\frac d{dt}g_{1i}^{(p)}(t)\big|_{t=t_0}=\lim_{\epsilon\to0}\frac{g_{1i}^{(p)}(t_0)-g_{1i}^{(p)}(t_0-\epsilon)}{\epsilon}
\leq\lim_{\epsilon\to0}\frac{g_{2i}(t_0)-g_{2i}(t_0-\epsilon)}{\epsilon}= \frac d{dt}g_{2i}(t)\big|_{t=t_0}.
\end{align*}
On the other hand, as $f(t,\cdot)$ satisfies the type $K$ condition for each $t\in[0,T]$ and $\tilde{f}(t,x)\geq0$ for all $(t,x)\in[0,T]\times\R^n$, for the above $i$, we also have that
\begin{align}
\frac d{dt}g_{1i}^{(p)}(t)\big|_{t=t_0}=&f_i(t_0,g_{1i}^{(p)}(t_0))+\tilde{f}_i(t_0,g_{1}^{(p)}(t_0))+\frac1p\nonumber\\
>&f_i(t_0,g_{1i}^{(p)}(t_0))\geq f_i(t_0,g_2(t_0))=\frac d{dt}g_{2i}(t)\big|_{t=t_0}.
\end{align}
We obtain a contradiction, and hence $g_{1}^{(p)}(t)\gg g_2(t)$ for all $t\in[0,T]$. It therefore holds that $g_1(t)\geq g_2(t)$ for all $t\in[0,T]$ by passing $p$ to infinity. \hfill$\Box$\\

Now we are ready to present the following existence and uniqueness result for the positive (classical) solution of Eq.~\eqref{eq:hjbn-1}.
\begin{theorem}\label{thm:solutionk}
For each $k=0,1,\ldots,N-1$, assume that DPE~\eqref{eq:hjbeqn} admits a positive unique (classical) solution $\varphi^{(k+1),j}(t)$ on $t\in[0,T]$ for $j\notin\{j_1,\ldots,j_k\}$. Then, there exists a unique positive (classical) solution $\varphi^{(k)}(t)$ on $t\in[0,T]$ of \eqref{eq:hjbeqn}
at the default state $z=0^{j_1,\ldots,j_k}$ (i.e., Eq.~\eqref{eq:hjbn-1} admits a unique positive (classical) solution).
\end{theorem}

\noindent{\it Proof.}\quad For any constant $a\in(0,1]$, let us consider the truncated dynamical system given by
\begin{align}\label{eq:truneqn}
\left\{
\begin{aligned}
\frac d{dt} \varphi_a^{(k)}(t)+A^{(k)}\varphi_a^{(k)}(t) + G_a^{(k)}(t,\varphi_a^{(k)}(t))=&0,\ \text{ in }[0,T);\\
\varphi^{(k)}_a(T)=&e_n.
\end{aligned}
\right.
\end{align}
Here $\varphi_a^{(k)}(t)=(\varphi_a^{(k)}(t,i);\ i=1,\ldots,n)^{\top}$ is the vector-valued solution and the $n\times n$-dimensional matrix $A^{(k)}$ is given by \eqref{eq:An-1}. The vector-valued function $G_a^{(k)}(t,x)$ is defined as:
\begin{align}\label{eq:Ga}
G_a^{(k)}(t,x) := G^{(k)}(t,x\vee a e_n),\qquad (t,x)\in[0,T]\times\R^n.
\end{align}
Thanks to Lemma~\ref{lem:Gkesti}, there exists a positive constant $C=C(a)$ which only depends on $a>0$ such that, for all $t\in[0,T]$,
\begin{align}\label{eq:Lip-Ga}
\big\|G_a^{(k)}(t,x)-G_a^{(k)}(t,y)\big\|\leq C\|x-y\|,\qquad x,y\in\R^n,
\end{align}
i.e., $G^{(k)}_a(t,x)$ is globally Lipschitz continuous w.r.t. $x\in\R^m$ uniformly in $t\in[0,T]$. By reversing the time, let us consider $\tilde{\varphi}_a^{(k)}(t):=\varphi_a^{(k)}(T-t)$ for $t\in[0,T]$. Then, $\tilde{\varphi}_a^{(k)}(t)$ satisfies the following dynamical system given by
\begin{align}\label{eq:truneq2}
\left\{
\begin{aligned}
\frac{d}{dt}\tilde{\varphi}_a^{(k)}(t)=&A^{(k)}\tilde{\varphi}^{(k)}_a(t)+G^{(k)}_{a}(T-t,\tilde{\varphi}_a^{(k)}(t)),\ \text{ in }(0,T];\\
\tilde{\varphi}_a^{(k)}(0)=&e_n^{\top}.
\end{aligned}
\right.
\end{align}
In virtue of the globally Lipschitz continuity condition \eqref{eq:Lip-Ga}, for each $a\in(0,1]$, it follows that the system~\eqref{eq:truneq2} has a unique (classical) solution $\tilde{\varphi}_a^{(k)}(t)$ on $[0,T]$.
In order to apply Lemma~\ref{comparison}, we rewrite the above system \eqref{eq:truneq2} in the following form:
\begin{align}\label{eq:truneq3}
\left\{
\begin{aligned}
\frac{d}{dt}\tilde{\varphi}_a^{(k)}(t)=&f^{(k)}(\tilde{\varphi}^{(k)}_a(t))+\tilde{f}_a^{(k)}(t,\tilde{\varphi}_a^{(k)}(t)),\ \text{ in }(0,T];\\
\tilde{\varphi}_a^{(k)}(0)=&e_n.
\end{aligned}
\right.
\end{align}
Here, the Lipschitz continuous functions $f^{(k)}(x)=(f_i^{(k)}(x);\ i=1,\ldots,n)^{\top}$ and $\tilde{f}_a^{(k)}(t,x)=(\tilde{f}^{(k)}_{a,i}(t,x);\ i=1,\ldots,n)^{\top}$ on $(t,x)\in[0,T]\times\R^n$ are given respectively by
\begin{align}\label{eq:f}
f_i^{(k)}(x)&=\sum_{j=1}^nq_{ij}x_j-\left(\frac{\theta}{2} r(i)+\sum_{j\notin\{j_1,\ldots,j_{k}\}}h_{j}^{(k)}(i)\right)x_i-\beta_i\{|x_i|\vee1\},\nonumber\\
\tilde{f}_{a,i}^{(k)}(t,x)&=G_a^{(k)}(T-t,x)+\beta_i\{|x_i|\vee1\},\quad i=1,\ldots,n.
\end{align}
The positive constants $\beta_i$ for $i\in D_n$ are given by
\begin{align}\label{eq:betai}
\beta_i=&-\inf_{\pi^{(k)}\in U^{(k)}}H^{(k)}(\pi^{(k)};i),
\end{align}
where, for $i\in D_n$, $H^{(k)}(\pi^{(k)};i)$ is defined by \eqref{eq:Hk}. It is not difficult to see that $\beta_i$ is a nonnegative and finite constant for each $i\in D_n$ using \eqref{eq:Hk}. By the recursive assumption that $\varphi^{(k+1),j}(t)\gg0$ on $[0,T]$ for $j\notin\{j_1,\ldots,j_k\}$, for any $a\in(0,1]$, we have that, for each $i\in D_n$, and all $(t,x)\in[0,T]\times\R^n$,
\begin{equation}\label{eq:Gapositive}
\begin{split}
&G^{(k)}_i(T-t,x\vee ae_n)\\
=&\inf_{\pi^{(k)}\in U^{(k)}}\left\{\sum_{j\notin\{j_1,\ldots,j_k\}}\varphi^{(k+1),j}(T-t,i)(1-\pi_{j}^{(k)})^{-\frac\theta2}\lambda_{j}^{(k)}(i)+H^{(k)}(\pi^{(k)};i)(x_i\vee a)\right\}\\
\geq&\{x_i\vee a\}\inf_{\pi^{(k)}\in U^{(k)}}H^{(k)}(\pi^{(k)};i)\geq-\beta_i\{|x_i|\vee 1\}.
\end{split}
\end{equation}
Thus, from \eqref{eq:f}, it follows that, for all $(t,x)\in[0,T]\times\R^n$,
\begin{align}\label{eq:onftilde}
\tilde{f}_{a,i}^{(k)}(t,x)=G^{(k)}_i(T-t,x\vee ae_n)+\beta_i\{|x_i|\vee1\}\geq0,\quad i\in D_n.
\end{align}

We next verify that the vector-valued function $f^{(k)}(x)=(f_i^{(k)}(x);\ i=1,\ldots,n)^{\top}$ given by \eqref{eq:f} is of type $K$. Namely we need to verify that, for any $x,y\in\R^n$ satisfying $x\leq y$ and $x_{i_0}=y_{i_0}$ for some $i_0=1,\ldots,n$, it holds that $f_{i_0}^{(k)}(x)\leq f_{i_0}^{(k)}(y)$. In fact, by \eqref{eq:f}, we have that, for any $x,y\in\R^n$ satisfying $x\leq y$ and $x_{i_0}=y_{i_0}$ for some $i_0=1,\ldots,n$,
\begin{align}\label{eq:condK}
f_{i_0}^{(k)}(x)&=\sum_{j=1}^nq_{i_0j}x_j-\left(\frac{\theta}{2} r(i_0)+\sum_{j\notin\{j_1,\ldots,j_{k}\}}\lambda_{j}^{(k)}(i_0)\right)x_{i_0}-\beta_i\{|x_{i_0}|\vee1\}\nonumber\\
&=q_{i_0i_0}x_{i_0}-\left(\frac{\theta}{2} r(i_0)+\sum_{j\notin\{j_1,\ldots,j_{k}\}}\lambda_{j}^{(k)}(i_0)\right)x_{i_0}-\beta_{i_0}\{|x_{i_0}|\vee1\}+\sum_{j\neq i_0}q_{i_0j}x_j\nonumber\\
&=q_{i_0i_0}y_{i_0}-\left(\frac{\theta}{2} r(i_0)+\sum_{j\notin\{j_1,\ldots,j_{k}\}}\lambda_{j}^{(k)}(i_0)\right)y_{i_0}-\beta_{i_0}\{|y_{i_0}|\vee1\}+\sum_{j\neq i_0}q_{i_0j}x_j\nonumber\\
&\leq q_{i_0i_0}y_{i_0}-\left(\frac{\theta}{2} r(i_0)+\sum_{j\notin\{j_1,\ldots,j_{k}\}}\lambda_{j}^{(k)}(i_0)\right)y_{i_0}-\beta_{i_0}\{|y_{i_0}|\vee1\}+\sum_{j\neq i_0}q_{i_0j}y_j\nonumber\\
&=f_{i_0}^{(k)}(y),
\end{align}
where we used the fact that for all $j\neq i_0$, $q_{i_0j}\geq0$ as $Q_n=(q_{ij})_{n\times n}$ is the generator of the Markov chain $Y$ and hence $\sum_{j\neq i_0}q_{i_0j}x_j\leq \sum_{j\neq i_0}q_{i_0j}y_j$ for all $x\leq y$. Hence, using Proposition 1.1 of Chapter 3 in Smith \cite{smith08} and Lemma~\ref{lem:sol-hjben2}, we deduce that the following dynamical system
\begin{align}\label{eq:truneq4}
\left\{
\begin{aligned}
\frac{d}{dt}{\psi}^{(k)}(t)=&f^{(k)}({\psi}^{(k)}(t)),\ \text{ in }(0,T];\\
{\psi}^{(k)}(0)=&e_n
\end{aligned}
\right.
\end{align}
admits a unique (classical) solution ${\psi}^{(k)}(t)=(\psi_i^{(k)}(t);\ i=1,\ldots,n)^{\top}$ on $t\in[0,T]$, and moreover it holds that ${\psi}^{(k)}(t)\gg0$ for $t\in[0,T]$. Let us set
\begin{align}\label{eq:epsilonk}
\varepsilon^{(k)}:=\min_{i=1,\ldots,n}\left\{\inf_{t\in[0,T]}\psi_i^{(k)}(t)\right\}.
\end{align}
The continuity of $\psi^{(k)}(t)$ in $t\in[0,T]$ and $\psi^{(k)}(t)\gg0$ for all $t\in[0,T]$ lead to $\varepsilon^{(k)}>0$. On the other hand, it follows from \eqref{eq:onftilde} that
the vector-valued function $f_a^{(k)}(t,x)\geq0$ on $[0,T]\times\R^n$. Because the vector-valued function $f^{(k)}(x)$ is also of type $K$ proved by \eqref{eq:condK}, we can apply Lemma~\ref{comparison} to the dynamical systems \eqref{eq:truneq3} and \eqref{eq:truneq4} and derive that
\begin{align}\label{eq:comparison0}
\tilde{\varphi}_a^{(k)}(t)\geq {\psi}^{(k)}(t)\geq\varepsilon^{(k)}e_n,\qquad \forall\ t\in[0,T],
\end{align}
as $\tilde{\varphi}_a^{(k)}(0)={\psi}^{(k)}(0)=e_n$. Note that the positive constant $\varepsilon^{(k)}$ given by \eqref{eq:epsilonk} above is independent of the constant $a\in(0,1]$. We can therefore choose $a\in(0,\varepsilon^{(k)}\wedge1)$ and it holds that $G_a^{(k)}(T-t,\tilde{\varphi}_a^{(k)}(t))=G^{(k)}(T-t,\tilde{\varphi}_a^{(k)}(t)\vee ae_n)=G^{(k)}(T-t,\tilde{\varphi}_a^{(k)}(t))$ on $[0,T]$. By \eqref{eq:truneq2} with $a\in(0,\varepsilon^{(k)}\wedge1)$, it follows that $\tilde{\varphi}_a^{(k)}(t)\geq\varepsilon^{(k)}e_n$ on $[0,T]$ is the unique (classical) solution to the dynamical system \eqref{eq:hjbn-1} and the proof of the theorem is complete. \hfill$\Box$\\

As an important implication of Theorem~\ref{thm:solutionk}, we present one of our major contributions to the existing literature in the next proposition as the characterization of the optimal strategy $\pi^{(k)}\in{U}^{(k)}$ at the default state $z=0^{j_1,\ldots,j_k}$ where $k=0,1,\ldots,N-1$.
\begin{proposition}\label{coro:optimal-strategy}
For each $k=0,1,\ldots,N-1$, assume that DPE~\eqref{eq:hjbeqn} admits a positive unique (classical) solution $\varphi^{(k+1),j}(t)$ on $t\in[0,T]$ for $j\notin\{j_1,\ldots,j_k\}$. Let $\varphi^{(k)}(t)=(\varphi^{(k)}(t,i);\ i=1,\ldots,n)^{\top}$ be the unique (classical) solution of DPE \eqref{eq:hjbn-1}. Then, there exists a unique optimal feedback strategy $\pi^{(k,*)}=\pi^{(k,*)}(t,i)$ for $(t,i)\in[0,T]\times D_n$ which is given explicitly by
\begin{align}\label{eq:optimal-strategy}
\pi^{(k,*)}=&\pi^{(k,*)}(t,i)\\
=&\argmin_{\pi^{(k)}\in U^{(k)}}\left\{\sum_{j\notin\{j_1,\ldots,j_k\}} \varphi^{(k+1),j}(t,i)\big(1-\pi_{j}^{(k)}\big)^{-\frac{\theta}{2}}\lambda_{j}^{(k)}(i)+H^{(k)}(\pi^{(k)};i)\varphi^{(k)}(t,i)\right\}\nonumber\\
=&\argmin_{\substack{\pi^{(k)}\in\{\pi\in U^{(k)}:\\ \|\pi\|\leq C\}}}\Bigg\{\sum_{j\notin\{j_1,\ldots,j_k\}} \varphi^{(k+1),j}(t,i)\big(1-\pi_{j}^{(k)}\big)^{-\frac{\theta}{2}}\lambda_{j}^{(k)}(i)+H^{(k)}(\pi^{(k)};i)\varphi^{(k)}(t,i)\Bigg\},\nonumber
\end{align}
for some positive constant $C>0$.
\end{proposition}

\noindent{\it Proof.}\quad Let us first recall Eq.~\eqref{eq:hjbn-1}, i.e.,
\begin{align*}
\left\{
\begin{aligned}
\frac d{dt}\varphi^{(k)}(t)=&-A^{(k)}\varphi^{(k)}(t)-G^{(k)}(t,\varphi^{(k)}(t)),\quad\text{ in }[0,T);\\
\varphi^{(k)}(T)=&e_n.
\end{aligned}
\right.
\end{align*}
Theorem~\ref{thm:solutionk} above shows that the above dynamical system admits a unique positive (classical) solution $\varphi^{(k)}(t)$ on $[0,T]$ and moreover $\varphi^{(k)}(t)\geq \varepsilon^{(k)}e_n^{\top}$ for all $t\in[0,T]$. Here $\varepsilon^{(k)}>0$ is given by~\eqref{eq:epsilonk}. Thus, by \eqref{eq:G2}, we have that, there exists a positive constant $C(\varepsilon^{(k)})$ which depends on $\varepsilon^{(k)}>0$ such that, for each $i\in D_n$,
\begin{align*}
&G^{(k)}_i(t,\varphi^{(k)}(t,i))\nonumber\\
=&\inf_{\substack{\pi^{(k)}\in\{\pi\in U^{(k)}:\\ \|\pi\|\leq C(\varepsilon^{(k)})\}}}\Bigg\{\sum_{j\notin\{j_1,\ldots,j_k\}} \varphi^{(k+1),j}(t,i)(1-\pi_{j}^{(k)})^{-\frac{\theta}{2}}\lambda_{j}^{(k)}(i)+H^{(k)}(\pi^{(k)};i)\varphi^{(k)}(t,i)\Bigg\}.\nonumber
\end{align*}
Here, for each $i=1,\ldots,n$, the function $G_i^{(k)}(t,x)$ on $(t,x)\in[0,T]\times\R^n$ is given by \eqref{eq:Gin-1}.
Also for each $i=1,\ldots,n$, $\varphi^{(k+1),j}(t,i)$ on $t\in[0,T]$ is the $i$-th element of the positive (classical) solution $\varphi^{(k+1),j}(t)$ of \eqref{eq:hjbeqn} at the default state $z=0^{j_1,\ldots,j_k,j}$ for $j\notin\{j_1,\ldots,j_k\}$. Recall that the function $H^{(k)}(\pi^{(k)};i)$ for $(\pi^{(k)},i)\in U^{(k)}\times D_n$ is given by \eqref{eq:Hk}. Then, it is not difficult to see that, for each $i\in D_n$ and fixed $t\in[0,T]$,
\[
h^{(k)}(\pi^{(k)},i):=\sum_{j\notin\{j_1,\ldots,j_k\}} \varphi^{(k+1),j}(t,i)(1-\pi_{j}^{(k)})^{-\frac{\theta}{2}}\lambda_{j}^{(k)}(i)+H^{(k)}(\pi^{(k)};i)\varphi^{(k)}(t,i)
\]
is continuous and strictly convex in $\pi^{(k)}\in\bar{U}^{(k)}$. Also notice that the space $\{\pi^{(k)}\in \bar{U}^{(k)};\ \|\pi^{(k)}\|\leq C(\varepsilon^{(k)})\}\subset\R^{N-k}$ is compact.
Hence, there exist a unique optimum $\pi^{(k,*)}=\pi^{(k,*)}(t,i)\in\bar{U}^{(k)}$. Moreover, it is noted that $h^{(k)}(\pi^{(k)},i)=+\infty$ when $\pi^{(k)}\in\bar{U}^{(k)}\setminus U^{(k)}$ while $h^{(k)}(\pi^{(k)},i)<+\infty$ for all $\pi^{(k)}\in U^{(k)}$. Consequently, we in fact obtain the optimum $\pi^{(k,*)}=\pi^{(k,*)}(t,i)\in\bar{U}^{(k)}$ admitting the representation \eqref{eq:optimal-strategy} by taking $C=C(\varepsilon^{(k)})$ which completes the proof of the proposition. \hfill$\Box$\\

As one of our main results, we finally present and prove the verification theorem for the finite state space of the regime-switching process $Y$ in the next proposition.
\begin{proposition}\label{prop:verithemfinite}
Let $\varphi(t,z)=(\varphi(t,i,z);\ i\in D_n)^{\top}$ with $(t,z)\in[0,T]\times{\cal S}$ be the unique solution of DPE~\eqref{eq:hjbeqn}. For $(t,i,z)\in[0,T]\times D_n\times{\cal S}$, define
\begin{align}\label{eq:pistar}
\pi^*(t,i,z):={\rm diag}((1-z_j)_{j=1}^N)\argmin_{\pi\in U}\tilde{H}\left(\pi;i,z,(\varphi(t,i,z^j);\ j=0,1,\ldots,N)\right),
\end{align}
where $\tilde{H}(\pi;i,z,\bar{f}(z))$ is given by \eqref{eq:H}. Let $\tilde{\pi}^*=(\tilde{\pi}^*(t))_{t\in[0,T]}$ with $\tilde{\pi}^*(t):=\pi^*(t,Y(t-),Z(t-))$. Then $\tilde{\pi}^*\in\tilde{\cal U}$ and it is the optimal feedback strategy, i.e., it holds that
\begin{align}\label{optimeq}
-\frac{2}{\theta}\log\Ex_{t,i,z}^{\tilde{\pi}^*,\theta}\left[\exp\left(\frac{\theta}{2}\int_t^TL(\tilde{\pi}^*(s);Y(s),Z(s))ds\right)\right]=\bar{V}(t,i,z)=-\frac2\theta\log\varphi(t,i,z).
\end{align}
\end{proposition}

\begin{proof} From Proposition~\ref{coro:optimal-strategy}, it follows that $\tilde{\pi}^*$ is a bounded and predictable process taking values on $U$. We next prove that $\tilde{\pi}^*$ is uniformly away from $1$. In fact, for fixed $(i,z,x)\in D_n\times\mathcal{S}\times(0,\infty)^{N+1}$, we have that $\tilde{H}\left(\pi;i,z,x\right)$ is strictly convex w.r.t. $\pi\in U$, thus $\Phi(i,z,x):=\argmin_{\pi\in U}\tilde{H}\left(\pi;i,z,x\right)$ is well-defined. Notice that $\Phi(i,z,\cdot)$ maps $(0,\infty)^{N+1}$ to $U$ and satisfies the first-order condition $\frac{\partial\tilde{H}}{\partial\pi_j}\left(\Phi(i,z,x);i,z,x\right)=0$ for $j=1,\ldots,N$.
Then, Implicit Function Theorem yields that $\Phi(i,z,x)$ is continuous in $x$. Further, for $j=1,\ldots,N,$ if $Z_j(t-)=0$, the first-order condition gives that
\begin{align}\label{eq:pistarbelow}
(1-\tilde\pi^*_j(t))^{-\frac\theta2-1}=&\bigg[\big(\mu_j(Y(t-))-r(Y(t-))\big)-\frac\theta2\left(1+\frac\theta2\right)\sum_{i=1}^N\big(\sigma^\top(Y(t-))\sigma(Y(t-))\big)_{ji}\tilde\pi^*_i(t)\nonumber\\
&+\frac\theta2\lambda_j(Y(t-),Z(t-))\bigg]\frac{\varphi(t,Y(t-),Z(t-))}{\lambda_j(Y(t-),Z(t-))\varphi(t,Y(t-),Z^j(t-))}.
\end{align}
Because for all $(i,z)\in D_n\times{\cal S}$, $\varphi(\cdot,i,z)$ has a strictly positive lower bound using \eqref{eq:comparison0}. Together with Proposition~\ref{coro:optimal-strategy}, it follows that, there exists a constant $C>0$ such that $\sup_{t\in[0,T]}(1-\tilde\pi^*_j(t))^{-\frac\theta2-1}\leq C$ for all $j=1,\ldots,N$. Hence, the estimate~\eqref{eq:pistarbelow} yields that
$\tilde{\pi}^*$ is uniformly bounded away from $1$. Thus, the following generalized Novikov's condition holds:
\begin{align}\label{eq:integral-cond}
\Ex\left[\exp\left(\frac{\theta^2}{8}\int_0^T\left|\sigma(Y(t))^{\top}\tilde{\pi}^*(t)\right|^2dt+\sum_{j=1}^N\int_0^T\left|(1-\tilde{\pi}_j^*(t))^{-\frac{\theta}{2}}-1\right|^2\lambda_j(Y(t),Z(t))dt\right)\right]<+\infty.
\end{align}
The above Novikov's condition \eqref{eq:integral-cond} implies that $\tilde{\pi}^*$ is admissible. We next prove \eqref{optimeq}. Noting that $\varphi(t,z)=(\varphi(t,i,z);\ i\in D_n)^{\top}$ with $(t,z)\in[0,T]\times{\cal S}$ is the unique classical solution of \eqref{eq:hjbeqn}. Note that, there exists a constant $C_L=C_L(n,i,z)>0$ such that $L(\pi;i,z)>-C_L$ for $(\pi,i,z)\in U\times D_n\times{\cal S}$. For $m\geq1$, set $L_m(\pi;i,z):=L(\pi;,i,z)\wedge m$. Then $L_m$ is bounded and $L_m(\pi;i,z)\uparrow L(\pi;i,z)$ as $m\to\infty$. Therefore, for any admissible strategy $\tilde{\pi}\in\tilde{\cal U}$, It\^o's formula gives that, for $0\leq t<s\leq T$,
\begin{align}\label{eq:itoveri0}
&\Ex_{t,i,z}^{\tilde{\pi},\theta}\left[\varphi(s,Y(s),Z(s))\exp\left(\frac{\theta}{2}\int_t^sL_m(\tilde{\pi}(u);Y(u),Z(u))du\right)\right]\nonumber\\
&\quad =\varphi(t,i,z)+\Ex_{t,i,z}^{\tilde{\pi},\theta}\Bigg[\int_t^s\exp\left(\frac{\theta}{2}\int_t^uL_m(\tilde{\pi}(v);Y(v),Z(v))dv\right)\nonumber\\
&\quad\qquad\times\Bigg\{\frac{\partial\varphi(u,Y(u),Z(u))}{\partial t}+\sum_{l\neq Y(u)}q_{Y(u)l}\left(\varphi(u,l,Z(u))-\varphi(u,Y(u),Z(u))\right)\nonumber\\
&\qquad\qquad\quad+\tilde{H}\left(\tilde{\pi}(u);Y(u),Z(u),(\varphi(t,Y(u),Z^j(u));\ j=0,1,\ldots,N)\right)\Bigg\}du\Bigg]\nonumber\\
&\qquad\quad+\Ex_{t,i,z}^{\tilde{\pi},\theta}\Bigg[\int_t^s\exp\left(\frac{\theta}{2}\int_t^uL_m(\tilde{\pi}(v);Y(v),Z(v))dv\right)\varphi(u,Y(u),Z(u))\nonumber\\
&\qquad\qquad\qquad\qquad\times(L_m-L)(\tilde{\pi}(u);Y(u),Z(u))du\Bigg]\nonumber\\
&\quad\geq\varphi(t,i,z)+\Ex_{t,i,z}^{\tilde{\pi},\theta}\Bigg[\int_t^s\exp\left(\frac{\theta}{2}\int_t^uL_m(\tilde{\pi}(v);Y(v),Z(v))dv\right)\varphi(u,Y(u),Z(u))\nonumber\\
&\qquad\qquad\qquad\qquad\times(L_m-L)(\tilde{\pi}(u);Y(u),Z(u))du\Bigg].
\end{align}
In the last inequality above, the integral term in the expectation is negative. On the other hand, note that $\varphi$ is bounded and positive, this integral also admits that, $\Px_{t,i,z}^{\tilde{\pi},\theta}$-a.s., for some constant $C_{\varphi}>0$,
\begin{align*}
&\int_t^s\exp\left(\frac{\theta}{2}\int_t^uL_m(\tilde{\pi}(v);Y(v),Z(v))dv\right)\varphi(u,Y(u),Z(u))(L_m-L)(\tilde{\pi}(u);Y(u),Z(u))du\nonumber\\
&\quad\geq-C_{\varphi}\int_t^s\exp\left(\frac{\theta}{2}\int_t^u[L(\tilde{\pi}(v);Y(v),Z(v))+C_L]dv\right)[L(\tilde{\pi}(u);Y(u),Z(u))+C_L]du\nonumber\\
&\quad=\frac{2 C_{\varphi}}{\theta}\left[1-e^{\frac{\theta}{2}C_L(s-t)}\exp\left(\frac{\theta}{2}\int_t^sL(\tilde{\pi}(u);Y(u),Z(u))du\right)\right].
\end{align*}
By taking $s=T$ above. Then, from Dominated Convergence Theorem, it follows that
\begin{align}\label{eq:itoveri}
\Ex_{t,i,z}^{\tilde{\pi},\theta}\left[\varphi(T,Y(T),Z(T))\exp\left(\frac{\theta}{2}\int_t^TL(\tilde{\pi}(u);Y(u),Z(u))du\right)\right]\geq\varphi(t,i,z).
\end{align}
Note that $\varphi(T,i,z)=1$ in \eqref{eq:itoveri}, we obtain that
\begin{align}\label{infoverphi}
\inf_{\tilde{\pi}\in\tilde{\cal U}}\Ex_{t,i,z}^{\tilde{\pi},\theta}\left[\exp\left(\frac{\theta}{2}\int_t^TL(\tilde{\pi}(u);Y(u),Z(u))du\right)\right]\geq\varphi(t,i,z).
\end{align}
On the other hand, from \eqref{eq:itoveri0} and \eqref{eq:pistar}, it follows that, for $0\leq t<s\leq T$,
\begin{align}\label{infoverphi2}
\Ex_{t,i,z}^{\tilde{\pi}^*,\theta}\left[\exp\left(\frac{\theta}{2}\int_t^TL(\tilde{\pi}^*(u);Y(u),Z(u))du\right)\right]=\varphi(t,i,z).
\end{align}
Because $\pi^*$ is admissible, i.e., $\tilde{\pi}^*\in\tilde{\cal U}$, we deduce from \eqref{infoverphi2} that
\begin{align}\label{phioverinf}
\varphi(t,i,z)\geq\inf_{\tilde{\pi}\in\tilde{\cal U}}\Ex_{t,i,z}^{\tilde{\pi},\theta}\left[\exp\left(\frac{\theta}{2}\int_t^TL(\tilde{\pi}(u);Y(u),Z(u))du\right)\right].
\end{align}
Combining \eqref{infoverphi} and \eqref{phioverinf}, we have that
\begin{align}\label{phi=inf}
\varphi(t,i,z)=\inf_{\tilde{\pi}\in\tilde{\cal U}}\Ex_{t,i,z}^{\tilde{\pi},\theta}\left[\exp\left(\frac{\theta}{2}\int_t^TL(\tilde{\pi}(u);Y(u),Z(u))du\right)\right].
\end{align}
The equality above is equivalent to $\varphi(t,i,z)=e^{-\frac\theta2\bar{V}(t,i,z)}$ due to \eqref{eq:J}. Hence, Eq.~\eqref{infoverphi2} together with \eqref{phi=inf} imply that \eqref{optimeq} holds, which ends the proof.
\end{proof}

\subsection{Countable State Case of Regime-Switching Process}

This section focuses on the existence of classical solutions to the original DPE~\eqref{eq:dpe3} and the corresponding verification theorem when the state space of the Markov chain
$Y$ is the countably infinite set $\Zx_+=\{1,2,\ldots\}$. The truncation method used in the finite state case fails to be applicable in the case $\Zx_+$. Instead, we shall establish a sequence of appropriately approximating risk sensitive control problems with finite state set $D_n^0:=D_n\cup\{0\}$ for $n\in\Zx_+$. Building upon the results in the finite state case in Section~\ref{sec:finite-states}, and by establishing valid uniform estimates, we can arrive at the desired conclusion that the smooth value functions corresponding to the above approximating control problems converge to the classical solution of \eqref{eq:dpe3} with countably infinite set $\Zx_+$ as $n$ goes to infinity.

Recall $D_n=\{1,2,\dots,n\}$ for the fixed $n\in\Zx_+$. We define the truncated counterpart of the regime-switching process $Y$ as: for $t\in[0,T]$,
\begin{align}\label{eq:Yn}
Y^{(n)}(t):=Y(t)\mathds{1}_{\{\tau_n>t\}},\qquad \tau_n^t:=\inf\{s\geq t;\ Y(s)\notin D_n\},
\end{align}
where $\tau_n:=\tau_n^0$ for $n\in\Zx_+$. By convention, we set $\inf\emptyset=+\infty$. Then, the process $Y^{(n)}=(Y^{(n)}(t))_{t\in[0,T]}$ is a continuous-time Markov chain with finite state space $D_n^0$. Here $0$ is understood as an absorbing state. The generator of $Y^{(n)}$ can therefore be given by the following $n+1$-dimensional square matrix:
\begin{align}\label{eq:An}
A_n:=\left[\begin{matrix}
     0      & 0      & \dots & 0 \\
     q^{(n)}_{10} & q_{11} & \dots & q_{1n} \\
     q^{(n)}_{20} & q_{21} & \dots & q_{2n} \\
         \vdots   & \vdots  & \vdots & \vdots     \\
     q^{(n)}_{n0} & q_{n1} & \dots & q_{nn}
\end{matrix}\right],
\end{align}
where $q^{(n)}_{m0}=-\sum_{i=1}^nq_{mi}=\sum_{{i\neq m,i>n}}q_{mi}$ for all $m\in D_n$. Thus, $Y^{(n)}$ is conservative. Here $q_{ij}$ for $i,j=1,\ldots,n$ are the same as given in Subsection~\ref{sub:RSP}.
Since $0$ is an absorbing state, we arrange values for the model coefficients at this state. More precisely, we set $r(0)=0$, $\mu(0)=0$, $\lambda(0,z)=\frac\theta2e_N^{\top}$ for all $z\in{\cal S}$, and $\sigma(0)\sigma(0)^\top=\frac4{2+\theta}I_{N}$. Here $I_N$ denotes the $N$-dimensional identity matrix.
Then, it follows from \eqref{eq:L0} and Taylor's expansion that $L(\pi;0,z)=\|\pi\|^2+\sum_{j=1}^N(1-z_j)[(1-\pi_j)^{-\frac{\theta}{2}}-1-\frac\theta2\pi_j]\geq0$
for all $(\pi,z)\in U\times{\cal S}$.

We next introduce the approximating risk-sensitive control problems where regime-switching processes take values on $D_n^0$. To this end, define $\tilde{\cal U}_n$ as the admissible control set $\tilde{\cal U}$, but the regime-switching process $Y$ is replaced with $Y^{(n)}$.
We then consider the following objective functional given by, for $\tilde{\pi}\in\tilde{\cal U}_n$ and $(t,i,z)\in[0,T]\times{D_n^0}\times{\cal S}$,
\begin{align}\label{eq:Jn00}
J_n(\tilde{\pi};t,i,z):=&\Ex_{t,i,z}^{\tilde{\pi},\theta}\left[\exp\left(\frac{\theta}{2}\int_t^{T\wedge{\tau_n^t}}L(\tilde{\pi}(s);Y(s),Z(s))ds\right)\right]\nonumber\\
=&\Ex_{t,i,z}^{\tilde{\pi},\theta}\left[\exp\left(\frac{\theta}{2}\int_t^{T\wedge{\tau_n^t}}L(\tilde{\pi}(s);Y^{(n)}(s),Z(s))ds\right)\right].
\end{align}
Here, the risk-sensitive cost function $L(\pi;i,z)$ for $(\pi,i,z)\in U\times\Zx_+\times{\cal S}$ is given by \eqref{eq:L0}. In order to apply the results in the finite state case obtained in Section~\ref{sec:finite-states}, we also need to propose the following objective functional given by, for $\tilde{\pi}\in\tilde{\cal U}_n$ and $(t,i,z)\in[0,T]\times{D_n^0}\times{\cal S}$,
\begin{align}\label{eq:Jn}
\tilde{J}_n(\tilde{\pi};t,i,z):=&\Ex_{t,i,z}^{\tilde{\pi},\theta}\left[\exp\left(\frac{\theta}{2}\int_t^{T}{L}(\tilde{\pi}(s);Y^{(n)}(s),Z(s))ds\right)\right].
\end{align}

We will consider the auxiliary value function defined by
\begin{align}\label{eq:Vnvalue}
V_n(t,i,z):=-\frac{2}{\theta}\inf_{\tilde{\pi}\in\tilde{\cal U}_n}\log\tilde{J}_n(\tilde{\pi};t,i,z),\qquad (t,i,z)\in[0,T]\times D_n^0\times{\cal S}.
\end{align}
We have the following characterization of the value function $V_n$ which will play an important role in the study of the convergence of $V_n$ as $n\to\infty$.
\begin{lemma}\label{lem:jn=tildeJn}
It holds that $V_n(t,i,z)=-\frac{2}{\theta}\inf_{\tilde{\pi}\in\tilde{\cal U}_n}\log J_n(\tilde{\pi};t,i,z)$ for $(t,i,z)\in[0,T]\times D_n^0\times{\cal S}$.
\end{lemma}

\begin{proof}
Using \eqref{eq:Jn00} and \eqref{eq:Jn}, we have that, for all $\tilde{\pi}\in\tilde{\mathcal{U}}_n$,
\begin{align*}
&\log\tilde{J}_n(\tilde{\pi};t,i,z)\\
=&\log\Ex_{t,i,z}^{\tilde{\pi},\theta}\left[\exp\left(\frac{\theta}{2}\int_t^{T}{L}(\tilde{\pi}(s);Y^{(n)}(s),Z(s))ds\right)\right]\nonumber\\
=&\log\Ex_{t,i,z}^{\tilde{\pi},\theta}\left[\exp\left(\frac{\theta}{2}\int_t^{T\wedge\tau^t_n}{L}(\tilde{\pi}(s);Y^{(n)}(s),Z(s))ds
+\frac{\theta}{2}\int_{T\wedge\tau^t_n}^T{L}(\tilde{\pi}(s);Y^{(n)}(s),Z(s))ds\right)\right]\nonumber\\
=&\log\Ex_{t,i,z}^{\tilde{\pi},\theta}\left[\exp\left(\frac{\theta}{2}\int_t^{T\wedge\tau^t_n}{L}(\tilde{\pi}(s);Y^{(n)}(s),Z(s))ds
+\frac{\theta}{2}\int_{T\wedge\tau^t_n}^T{L}(\tilde{\pi}(s);0,Z(s))ds\right)\right]\nonumber\\
%&=\log\Ex_{t,i,z}^{\tilde{\pi},\theta}\left[\exp\left(\frac{\theta}{2}\int_t^{T\wedge\tau^t_n}L(\tilde{\pi}(s);Y^{(n)}(s),Z(s))ds+\frac{\theta}{2}\int_{T\wedge\tau^t_n}^T\|\tilde{\pi}(s)\|^2ds\right)\right]\nonumber\\
\geq & \log\Ex_{t,i,z}^{\tilde{\pi},\theta}\left[\exp\left(\frac{\theta}{2}\int_t^{T\wedge\tau^t_n}L(\tilde{\pi}(s);Y^{(n)}(s),Z(s))ds\right)\right]\nonumber\\
=&\log J_n(\tilde{\pi};t,i,z)\geq\inf_{\tilde{\pi}\in\tilde{\mathcal{U}}_n}\log J_n(\tilde{\pi};t,i,z),
\end{align*}
where we used the positivity of ${L}({\pi};0,z)$ for all $(\pi,z)\in U\times{\cal S}$. As $\theta>0$, we obtain from \eqref{eq:Vnvalue} that
\begin{align}\label{V<J}
V_n(t,i,z)&\leq-\frac2\theta\inf_{\tilde{\pi}\in\tilde{\mathcal{U}}_n}\log J_n(\tilde{\pi};t,i,z).
\end{align}
On the other hand, for any $\tilde{\pi}\in\tilde{\mathcal{U}}_n$, define $\hat{\pi}(t)=\tilde{\pi}(t)\mathds{1}_{\{t\leq\tau_n\}}$ for $t\in[0,T]$. It is clear that $\hat{\pi}\in\tilde{\mathcal{U}}_n$, and it holds that $\Gam^{\hat{\pi},\theta}(t,T):=\frac{\Gam^{\hat{\pi},\theta}(T)}{\Gam^{\hat{\pi},\theta}(t)}
=\frac{\Gam^{\tilde{\pi},\theta}(T\wedge\tau^t_n)}{\Gam^{\tilde{\pi},\theta}(t)}
=:\Gam^{\tilde{\pi},\theta}(t,T\wedge\tau^t_n)$. Hence
\begin{align*}
\log J_n(\tilde{\pi};t,i,z)
&=\log\Ex_{t,i,z}\left[\Gam^{\tilde{\pi},\theta}(t,T)\exp\left(\frac{\theta}{2}\int_t^{T\wedge\tau^t_n}L(\tilde{\pi}(s);Y^{(n)}(s),Z(s))ds\right)\right]\nonumber\\
%&=\log\Ex_{t,i,z}\left[\Ex\left\{\Gam^{\tilde{\pi},\theta}(t,T)\exp\left(\frac{\theta}{2}\int_t^{T\wedge\tau^t_n}L(\tilde{\pi}(s);Y^{(n)}(s),Z(s))ds\right)|\mathcal{F}_{T\wedge\tau^t_n}\right\}\right]\nonumber\\
&=\log\Ex_{t,i,z}\left[\exp\left(\frac{\theta}{2}\int_t^{T\wedge\tau^t_n}L(\tilde{\pi}(s);Y^{(n)}(s),Z(s))ds\right)
\Ex\left[\Gam^{\tilde{\pi},\theta}(t,T)|\mathcal{F}_{T\wedge\tau^t_n}\right]\right]\nonumber\\
&=\log\Ex_{t,i,z}\left[\Gam^{\tilde{\pi},\theta}(t,T\wedge\tau^t_n)\exp\left(\frac{\theta}{2}\int_t^{T\wedge\tau^t_n}
L(\tilde{\pi}(s);Y^{(n)}(s),Z(s))ds\right)\right]\nonumber\\
&=\log\Ex_{t,i,z}^{\hat{\pi},\theta}\left[\exp\left(\frac{\theta}{2}\int_t^{T\wedge\tau^t_n}L(\hat{\pi}(s);Y^{(n)}(s),Z(s))ds\right)\right]\nonumber\\
%&=\log\Ex_{t,i,z}\left[\Gam^{\hat{\pi},\theta}(t,T)\exp\left(\frac{\theta}{2}\int_t^{T\wedge\tau^t_n}L(\hat{\pi}(s);Y^{(n)}(s),Z(s))ds+\int_{T\wedge\tau^t_n}^T0ds\right)\right]\nonumber\\
&=\log\Ex_{t,i,z}^{\hat{\pi},\theta}\left[\exp\left(\frac{\theta}{2}\int_t^{T\wedge\tau^t_n}L(\hat{\pi}(s);Y^{(n)}(s),Z(s))ds+\frac{\theta}{2}\int_{T\wedge\tau^t_n}^TL(0;0,Z(s))ds\right)\right]\nonumber\\
&=\log\Ex_{t,i,z}^{\hat{\pi},\theta}\left[\exp\left(\frac{\theta}{2}\int_t^{T}{L}(\hat{\pi}(s);Y^{(n)}(s),Z(s))ds\right)\right]\nonumber\\
&=\log\tilde{J}_n(\hat{\pi};t,i,z)\geq\inf_{\tilde{\pi}\in\tilde{\mathcal{U}}_n}\log\tilde{J}_n(\tilde{\pi};t,i,z).
\end{align*}
The above inequality and the arbitrariness of $\tilde{\pi}$ jointly give that
\begin{align}\label{J<V}
-\frac2\theta\inf_{\tilde{\pi}\in\tilde{\mathcal{U}}_n}\log J_n(\tilde{\pi};t,i,z)\leq V_n(t,i,z).
\end{align}
Then, the desired result follows by combining \eqref{V<J} and \eqref{J<V} above.
\end{proof}

Lemma~\ref{lem:jn=tildeJn} together with Theorem~\ref{thm:solutionk} and Proposition~\ref{prop:verithemfinite} in Section \ref{sec:finite-states} for the finite state space of $Y$ imply the following conclusion:
\begin{proposition}\label{prop:Vnmonotone00}
Let $n\in\Zx_+$. Recall the value function $V_n(t,i,z)$ defined by \eqref{eq:Vnvalue}. We define $\varphi_n(t,i,z):=\exp(-\frac\theta2V_n(t,i,z))$. Then $\varphi_n(t,i,z)$ is the unique solution of the recursive system of DPEs given by
\begin{align}\label{eq:dpe4}
0=&\frac{\partial \varphi_n(t,i,z)}{\partial t}+\sum_{l\neq i,1\leq l\leq n}q_{il}\left(\varphi_n(t,l,z)-\varphi_n(t,i,z)\right)+q^{(n)}_{i0}(\varphi_n(t,0,z)-\varphi_n(t,i,z))\nonumber\\
&+\inf_{\pi\in U}\tilde{H}\left(\pi;i,z,(\varphi_n(t,i,z^j);\ j=0,1,\ldots,N)\right),
\end{align}
where $(t,i,z)\in[0,T)\times D_n^0\times{\cal S}$ and the terminal condition is given by $\varphi_n(T,i,z)=1$ for all $(i,z)\in D_n^0\times{\cal S}$. Moreover, it holds that $\varphi_n(t,i,z)\in[0,1]$ and it is decreasing in $n$ for all $(t,i,z)\in[0,T]\times D_n^0\times{\cal S}$.
\end{proposition}

\begin{proof}
Notice that the state space of $Y^{(n)}$ is given by $D_n^0$ which is a finite set. By observing the definition of the value function $V_n$ given by \eqref{eq:Vnvalue},
we have that $\varphi_n(t,i,z)$ is the unique solution of the recursive system \eqref{eq:dpe4} by applying Theorem~\ref{thm:solutionk} and
Proposition~\ref{prop:verithemfinite} in Section \ref{sec:finite-states} for the regime-switching process with the finite state space. In order to verify that $\varphi_n(t,i,z)\in[0,1]$ and it is decreasing in $n$, it is sufficient to prove that $V_n(t,i,z)\geq0$ and it is nondecreasing in $n$.
Thanks to Lemma~\ref{lem:jn=tildeJn}, and $L(0,i,z)=-r(i)\leq0$ by \eqref{eq:L0}, also note that $\tilde{\pi}_0(t)\equiv0$ is admissible (i.e., $\tilde{\pi}_0\in\tilde{\cal U}_n$), then
\begin{align*}%\label{lowboundforV}
\inf_{\tilde{\pi}\in\tilde{\cal U}_n}\log J_n(\tilde{\pi};t,i,z)&\leq \log J_n(\tilde{\pi}_0;t,i,z)
=\log\Ex_{t,i,z}^{\tilde{\pi}_0,\theta}\left[\exp\left(\frac{\theta}{2}\int_t^{T\wedge\tau_n^t}L(0;Y(s),Z(s))ds\right)\right]\nonumber\\
&=\log\Ex_{t,i,z}^{\tilde{\pi}_0,\theta}\left[\exp\left(-\frac{\theta}{2}\int_t^{T\wedge\tau_n^t}r(Y(s))ds\right)\right]\leq0,
\end{align*}
as the interest rate process is nonnegative. This gives that $V_n(t,i,z)\geq0$ for all $(t,i,z)\in[0,T]\times{D_n^0}\times{\cal S}$. On the other hand, for any $\tilde{\pi}\in\tilde{\cal U}_n$, we define $\hat{\pi}(t):=\tilde{\pi}(t)\mathds{1}_{\{\tau_n\geq t\}}$ for $t\in[0,T]$.
It is clear that $\hat{\pi}\in\tilde{\cal U}_n\cap \tilde{\cal U}_{n+1}$. Recall the density process given by \eqref{eq:Gam}, we have that, for $\tilde{\pi},\hat{\pi}\in\tilde{\cal U}_n$,
\begin{align*}
\Gamma^{\tilde{\pi},\theta}&={\cal E}(\Pi^{\tilde{\pi},\theta}),\ \Pi^{\tilde{\pi},\theta}=-\frac{\theta}{2}\int_0^{\cdot}\tilde{\pi}(s)^{\top}\sigma(Y^{(n)}(s))dW(s)+\sum_{j=1}^N\int_0^{\cdot}\{(1-\tilde{\pi}_j(s))^{-\frac{\theta}{2}}-1\}dM_j(s);\nonumber\\
\Gamma^{\hat{\pi},\theta}&={\cal E}(\Pi^{\hat{\pi},\theta}),\ \Pi^{\hat{\pi},\theta}=-\frac{\theta}{2}\int_0^{\cdot}\hat{\pi}(s)^{\top}\sigma(Y^{(n)}(s))dW(s)+\sum_{j=1}^N\int_0^{\cdot}\{(1-\hat{\pi}_j(s))^{-\frac{\theta}{2}}-1\}dM_j(s).
\end{align*}
This shows that $\Gamma^{\tilde{\pi},\theta}(t\wedge\tau_n)=\Gamma^{\hat{\pi},\theta}(t)$ for $t\in[0,T]$. Then, we deduce from \eqref{eq:Jn00} that
\begin{align}
\log J_n(\tilde{\pi};t,i,z)=&\log\Ex_{t,i,z}^{\tilde{\pi},\theta}\left[\exp\left(\frac{\theta}{2}\int_t^{T\wedge \tau_n^t}L(\tilde{\pi}(s);Y(s),Z(s))ds\right)\right]\nonumber\\
\geq&\log\Ex_{t,i,z}^{\tilde{\pi},\theta}\left[\exp\left(\frac{\theta}{2}\int_t^{T\wedge \tau_n^t}L(\tilde{\pi}(s);Y(s),Z(s))ds+\frac{\theta}{2}\int_{T\wedge \tau_n^t}^{T\wedge \tau_{n+1}^t}L(0;Y(s),Z(s))\right)\right]\nonumber\\
=&\log\Ex_{t,i,z}^{\hat{\pi},\theta}\left[\exp\left(\frac{\theta}{2}\int_t^{T\wedge\tau_{n+1}^t}L(\hat{\pi}(s);Y(s),Z(s))ds\right)\right]\nonumber\\
=&\log J_{n+1}(\hat{\pi};t,i,z)\geq\inf_{\tilde{\pi}\in\tilde{\cal U}_{n+1}}\log J_{n+1}(\tilde{\pi};t,i,z).
\end{align}
Using \eqref{eq:Vnvalue} and Lemma~\ref{lem:jn=tildeJn}, it follows that $V_{n}(t,i,z)$ is nondecreasing in $n$ for fixed $(t,i,z)\in[0,T]\times D_n^0\times{\cal S}$. Thus, the conclusion of the proposition holds.
\end{proof}

By virtue of Proposition~\ref{prop:Vnmonotone00}, for any $(t,i,z)\in[0,T]\times\Zx_+\times{\cal S}$, we set
$V^*(t,i,z):=\lim_{n\to\infty}V_n(t,i,z)$. Then, it holds that
\begin{align}\label{eq:varphistar}
\lim_{n\to\infty}\varphi_n(t,i,z)=\exp\left(-\frac\theta2V^*(t,i,z)\right)=:\varphi^*(t,i,z).
\end{align}
On the other hand, from Eq.~\eqref{eq:Vnvalue}, it is easy to see that $\varphi_n(t,0,z)=1$ for all $(t,z)\in[0,T]\times{\cal S}$. Then, Eq.~\eqref{eq:dpe4} above can be rewritten as:
\begin{align}\label{eq:dpe5}
\frac{\partial \varphi_n(t,i,z)}{\partial t}=&-q_{ii}\varphi_n(t,i,z)-\sum_{l\neq i,1\leq l\leq n}q_{il}\varphi_n(t,l,z)-\sum_{l>n}q_{il}\nonumber\\
&-\inf_{\pi\in U}\tilde{H}\left(\pi;i,z,(\varphi_n(t,i,z^j);\ j=0,1,\ldots,N)\right).
\end{align}
In terms of \eqref{eq:H}, we can conclude that, for $(\pi;i,z)\in U\times\Zx_+\times{\cal S}$, $\tilde{H}(\pi;i,z,x)$ is concave in every component of $x\in[0,\infty)^{N+1}$, so is $\inf_{\pi\in U}\tilde{H}(\pi;i,z,x)$. We present the main result in this paper for the case of the countable state space.
\begin{theorem}\label{thm:existD}
Let $(t,i,z)\in[0,T]\times\Zx_+\times{\cal S}$. Then, the limit function $\varphi^*(t,i,z)$ given in \eqref{eq:varphistar} above is a classical solution of the original DPE~\eqref{eq:dpe3}, i.e., it holds that
\begin{align*}
0=&\frac{\partial \varphi^*(t,i,z)}{\partial t}+\sum_{l\neq i}q_{il}\left[\varphi^*(t,l,z)-\varphi^*(t,i,z)\right]+\inf_{\pi\in U}\tilde{H}\left(\pi;i,z,(\varphi^*(t,i,z^j);\ j=0,1,\ldots,N)\right)
\end{align*}
with terminal condition $\varphi^*(T,i,z)=1$ for all $(i,z)\in\Zx_+\times{\cal S}$.
\end{theorem}

The proof of Theorem~\ref{thm:existD} will be split into proving a sequence of auxiliary lemmas first. We show the following result as a preparation.
\begin{lemma}\label{lem:boundfordphi}
Let $(i,z)\in \Zx_+\times{\cal S}$. Then $(\frac{\partial\varphi_n(t,i,z)}{\partial t})_{n\geq i}$ is uniformly bounded in $t\in[0,T]$.
\end{lemma}

\begin{proof}
We rewrite Eq.~\eqref{eq:dpe5} as in the following form:
\begin{align}\label{estm1}
&\frac{\partial \varphi_n(t,i,z)}{\partial t}=-q_{ii}\varphi_n(t,i,z)-\sum_{l\neq i,1\leq l\leq n}q_{il}\varphi_n(t,l,z)-\sum_{l>n}q_{il}\nonumber\\
&\qquad-\inf_{\pi\in U}\hat{H}\left(\pi;i,z,(\varphi_n(t,i,z^j);\ j=0,1,\ldots,N)\right)+C(i,z)\varphi_n(t,i,z),
\end{align}
where, for $(i,z)\in\Zx_+\times{\cal S}$,
\begin{align}\label{estm2}
C(i,z)=\bigg|\inf_{\pi\in U}\bigg\{&-\frac{\theta}{2}r(i)-\frac{\theta}{2}\pi^{\top}(\mu(i)-r(i)e_n)+\frac{\theta}{4}\left(1+\frac{\theta}{2}\right)\left\|\sigma(i)^{\top}\pi\right\|^2\nonumber\\
&+\sum_{j=1}^N\left(-1-\frac{\theta}{2}\pi_j\right)(1-z_j)\lambda_j(i,z)\bigg\}\bigg|,
\end{align}
and the nonnegative function
\begin{align}\label{eq:hatH}
\hat{H}(\pi;i,z,\bar{f}(z)):=\tilde{H}(\pi;i,z,\bar{f}(z))+C(i,z)f(z).
\end{align}
Because $\hat{H}(\pi;i,z,x)$ is concave in every component of $x\in[0,\infty)^{N+1}$, $\Phi(x):=\inf_{\pi\in U}\hat{H}(\pi;i,z,x)$ is also concave in every component of $x\in[0,\infty)^{N+1}$. It follows from Proposition~\ref{prop:Vnmonotone00} that $x^{(n)}:=(\varphi_n(t,i,z^j);\ j=0,1,\ldots,N)\in[0,1]^{N+1}$. Using Lemma~\ref{lem:conbound}, there exits a constant $C>0$ which is independent of $x^{(n)}$ such that $0\leq \Phi(x^{(n)})\leq C$ for all $n\in\Zx_+$. Further, for fixed $(i,z)\in\Zx_+\times{\cal S}$,
\begin{align*}
&\left|-q_{ii}\varphi_n(t,i,z)-\sum_{l\neq i,1\leq l\leq n}q_{il}\varphi_n(t,l,z)-\sum_{l>n}q_{il}+C(i,z)\varphi_n(t,i,z)\right|
\leq-2q_{ii}+C(i,z).
\end{align*}
The desired result follows from Eq.~\eqref{estm1}.
\end{proof}

\begin{lemma}\label{lem:unfmconforphi}
Let $(i,z)\in\Zx_+\times{\cal S}$, then $(\varphi_n(t,i,z))_{n\geq i}$ (decreasingly) converges to $\varphi^*(t,i,z)$ uniformly in $t\in[0,T]$ as $n\to\infty$.
\end{lemma}

\begin{proof}
By Proposition~\ref{prop:Vnmonotone00}, Lemma~\ref{lem:boundfordphi}, and Azel\`a-Ascoli's Theorem, we have that $(\varphi_n(\cdot,i,z))_{n\geq i}$ contains an uniformly convergent subsequence. Moreover, Proposition~\ref{prop:Vnmonotone00} and \eqref{eq:varphistar} yield that $\varphi_n(t,i,z)$ (decreasingly) converges to $\varphi^*(t,i,z)$ uniformly in $t\in[0,T]$ as $n\to\infty$.
\end{proof}

\begin{lemma}\label{lem:phinlobnd}
Let $n\in\Zx_+$. Consider the following linear system: for $(t,i,z)\in(0,T]\times D_n^0\times{\cal S}$,
\begin{align}\label{eq:phin}
\frac{\partial\phi_n(t,i,z)}{\partial t}=&(q_{ii}-C(i,z))\phi_n(t,i,z)+\sum_{l\neq i,1\leq l\leq n}q_{il}\phi_n(t,l,z),\nonumber\\
\phi_n(0,i,z)=&1,
\end{align}
where $C(i,z)$ is given by \eqref{estm2}. Then, there exists a measurable function $\phi^*(t,i,z)$ such that $\phi_n(t,i,z)\nearrow\phi^*(t,i,z)$ as $n\to\infty$ for each fixed $(t,i,z)$. Moreover, it holds that $0<\phi_n(T-t,i,z)\leq\varphi_n(t,i,z)\leq1$ for $(t,i,z)\in[0,T]\times D_n^0\times{\cal S}$.
\end{lemma}

\begin{proof}
Let $(t,i,z)\in[0,T]\times D_n^0\times{\cal S}$ and define $g_n(t,i,z):=\varphi_n(T-t,i,z)$. It follows from Eq.~\eqref{estm1} that $g_n(\cdot,i,z)\in C^1((0,T])\cap C([0,T])$ for each fixed $(i,z)$ and satisfies that
\begin{align}\label{g_n}
\frac{\partial g_n(t,i,z)}{\partial t}=&(q_{ii}-C(i,z))g_n(t,i,z)+\sum_{l\neq i,1\leq l\leq n}q_{il}g_n(t,l,z)+\sum_{l>n}q_{il}\nonumber\\
&+Q(t,i,z,g_n(t,i,z)),\nonumber\\
g_n(0,i,z)=&1,
\end{align}
where $Q(t,i,z,x):=\inf_{\pi\in U}\hat{H}\left(\pi;i,z,x,g_n(t,i,z^1),\ldots,g_n(t,i,z^N)\right)$ for $x\in[0,\infty)$. We have from \eqref{eq:hatH} that $Q(t,i,z,x)\geq0$ for all $(t,x)\in[0,T]\times[0,\infty)$. Then $\sum_{l>n}q_{il}+Q(t,i,z,x)\geq0$. Note that the linear part of Eq.~\eqref{g_n} satisfies the $K$-type condition. Then, using the comparison result of Lemma~\ref{comparison}, it shows that $g_n(t,i,z)\geq\phi_n(t,i,z)$, and hence $\varphi_n(t,i,z)\geq\phi_n(T-t,i,z)$. Moreover, we deduce from Lemma \ref{lem:sol-hjben2} that $\phi_n(t,i,z)>0$. By virtue of Eq.~\eqref{eq:phin}, we have that $\phi_{n+1}(t,i,z)$ with $(t,i,z)\in[0,T]\times D_{n+1}^0\times{\cal S}$ satisfies that
\begin{equation}\label{eq:phin+1}
\begin{split}
\frac{\partial\phi_{n+1}(t,i,z)}{\partial t}=&(q_{ii}-C(i,z))\phi_{n+1}(t,i,z)+\sum_{l\neq i,1\leq l\leq n}q_{il}\phi_{n+1}(t,l,z)\\
&+q_{i,n+1}\phi_{n+1}(t,n+1,z),\nonumber\\
\phi_{n+1}(0,i,z)=&1.
\end{split}
\end{equation}
Because $q_{i,n+1}\phi_{n+1}(t,n+1,z)\geq0$ for $i\in D_n^0$, Lemma~\ref{comparison} shows that $\phi_{n+1}(t,i,z)\geq\phi_n(t,i,z)$ for all $(t,i,z)\in[0,T]\times{D_n^0}\times{\cal S}$. Therefore, there exists a measurable function $\phi^*(t,i,z)$ such that $\phi_n(t,i,z)\nearrow\phi^*(t,i,z)$ as $n\to\infty$ for each fixed $(t,i,z)\in[0,T]\times\Zx_+\times{\cal S}$.
\end{proof}

\begin{lemma}\label{lem:lobndphistar}
Let $(i,z)\in\Zx_+\times{\cal S}$. Then, there exists a positive constant $\delta=\delta(i,z)$ such that $\varphi^*(t,i,z)>\delta$ for all $t\in[0,T]$.
\end{lemma}

\begin{proof}
From Lemma~\ref{lem:phinlobnd}, we have that $\varphi_n(t,i,z)\geq\phi_n(T-t,i,z)$. Letting $n\rightarrow\infty$ and using Lemma~\ref{lem:unfmconforphi}, it follows that $\varphi^*(t,i,z)\geq\phi^*(T-t,i,z)\geq\phi_i(T-t,i,z)$. As $\phi_i(t,i,z)>0$ is continuous in $t\in[0,T]$, there exists a positive constant $\delta=\delta(i,z)$ such that $\inf_{t\in[0,T]}\phi_i(t,i,z)\geq\delta$. Therefore $\varphi^*(t,i,z)\geq\delta$ for all $t\in[0,T]$.
\end{proof}

We can finally conclude the proof of Theorem~\ref{thm:existD} using all previous results.

\noindent{\it Proof of Theorem~\ref{thm:existD}.}\quad We first prove that there exists a measurable function $\tilde{\varphi}(t,i,z)$ on $(t,i,z)\in[0,T]\times\Zx_+\times{\cal S}$ such that $\lim_{n\to\infty}\frac{\partial\varphi_n(t,i,z)}{\partial t}=\tilde{\varphi}(t,i,z)$ for $(t,i,z)\in[0,T]\times\Zx_+\times{\cal S}$. In fact, note that for $(t,i,z)\in[0,T]\times D_n^0\times{\cal S}$, $0\leq\varphi_{n+1}(t,i,z)\leq\varphi_{n}(t,i,z)\leq1$ for $n\in\Zx_+$. Then
\begin{align*}
\sum_{l\neq i,1\leq l\leq n}q_{il}\varphi_n(t,l,z)+\sum_{l>n}q_{il}\geq\sum_{l\neq i,1\leq l\leq n+1}q_{il}\varphi_{n+1}(t,l,z)+\sum_{l>n+1}q_{il}.
\end{align*}
This yields from \eqref{eq:varphistar} that $q_{ii}\varphi_n(t,i,z)\nearrow q_{ii}\varphi^*(t,i,z)$ as $n\to\infty$, and
\begin{align}\label{eq:conver11}
\sum_{l\neq i,1\leq l\leq n}q_{il}\varphi_n(t,l,z)+\sum_{l>n}q_{il}\searrow&\sum_{l\neq i,l\geq 1}q_{il}\varphi^*(t,l,z).
\end{align}
On the other hand, let $\Phi(x):=\inf_{\pi\in U}\tilde{H}(\pi;i,z,x)$ for $x\in[0,\infty)^{N+1}$. Then $\Phi(x):[0,\infty)^{N+1}\to\R$ is concave in every component of $x$.
Let $x^*(t):=(\varphi^*(t,i,z^j);\ j=0,1,\ldots,N)$ and $x^{(n)}(t):=(\varphi_n(t,i,z^j);\ j=0,1,\ldots,N)$ for $n\in\Zx_+$. Then $0\leq x^*(t)\leq x^{(n)}(t)$ for $n\in\Zx_+$ and $\lim_{n\to\infty}x^{(n)}(t)=x^*(t)$ using \eqref{eq:varphistar}. Moreover, Lemma~\ref{lem:lobndphistar} gives that $\delta\ll x^*\ll2$. It follows from Lemma~\ref{lem:conconver} that $\lim_{n\to\infty}\Phi(x^{(n)}(t))=x^*(t)$. Thus, by virtue of Eq.~\eqref{eq:dpe5}, as $n\to\infty$, one has
\begin{align}\label{eq:expresstildevarphi}
&\frac{\partial\varphi_n(t,i,z)}{\partial t}\to\tilde{\varphi}(t,i,z):=-q_{ii}\varphi^*(t,i,z)-\sum_{l\neq i,l\geq 1}q_{il}\varphi^*(t,l,z)-\Phi\left(x^*(t)\right).
\end{align}

We next prove that for $(i,z)\in\Zx_+\times{\cal S}$, $\frac{\partial\varphi_n(t,i,z)}{\partial t}\rightrightarrows\tilde{\varphi}(t,i,z)$ in $t\in[0,T]$ as $n\to\infty$. Here $\rightrightarrows$ denotes the uniform convergence. Eq.~\eqref{estm1} together with \eqref{eq:expresstildevarphi} first give that, for $(t,i,z)\in[0,T]\times D_n^0\times{\cal S}$,
\begin{align}\label{eq:I-II-III}
\frac{\partial \varphi_n(t,i,z)}{\partial t}-\tilde{\varphi}(t,i,z)&=\sum_{i=1}^3 B_i^{(n)}(t,i,z),
\end{align}
where
\begin{align}\label{eq:Bn}
B_1^{(n)}(t,i,z) &:= -q_{ii}(\varphi_n(t,i,z)-\varphi^*(t,i,z))+C(i,z)(\varphi_n(t,i,z)-\varphi^*(t,i,z)),\nonumber\\
B_2^{(n)}(t,i,z) &:= \sum_{l\neq i,1\leq l\leq n}q_{il}(\varphi_n(t,l,z)-\varphi^*(t,l,z))+\sum_{l>n}q_{il}(1-\varphi^*(t,i,z)),\nonumber\\
B_3^{(n)}(t,i,z) &:= \Phi(x^{(n)}(t))-\Phi(x^*(t)).
\end{align}
Here $\Phi(x):=\inf_{\pi\in U}\tilde{H}(\pi;i,z,x)$ for $x\in[0,\infty)^{N+1}$, $x^{(n)}(t):=(\varphi_n(t,i,z^j);\ j=0,1,\ldots,N)$, and $x^{*}(t):=(\varphi^*(t,i,z^j);\ j=0,1,\ldots,N)$.

Lemma~\ref{lem:unfmconforphi} guarantees that $\varphi_n(t,i,z)\rightrightarrows\varphi^*(t,i,z)$ in $t\in[0,T]$ as $n\to\infty$, and hence $B_1^{(n)}(t,i,z)\rightrightarrows0$ in $t\in[0,T]$ as $n\to\infty$.
On the other hand, for any small $\varepsilon>0$, since $\sum_{l\neq i}q_{il}<\infty$, there exists $n_1\geq1$ such that $\sum_{l>n_1,l\neq i}q_{il}<\frac\varepsilon2$. Note that, for all $1\leq l\leq n_1$, $\varphi_n(t,l,z)\rightrightarrows\varphi^*(t,l,z)$ in $t\in[0,T]$ as $n\to\infty$, there exists
$n_2\geq1$ such that $\sup_{t\in[0,T]}\sum_{l\neq i,1\leq l\leq n_1}q_{il}(\varphi_n(t,l,z)-\varphi^*(t,l,z))\leq\frac\varepsilon2$ for $n>n_2$. Hence, for all $n>n_1\vee n_2$, noting that $0\leq\varphi^*(t,i,z)\leq\varphi_n(t,i,z)\leq1$, it holds that
\begin{equation}\label{II}
\begin{split}
|B_2^{(n)}(t,i,z)|=&\sum_{l\neq i,1\leq l\leq n_1}q_{il}(\varphi_n(t,l,z)-\varphi^*(t,l,z))+\sum_{l\neq i,n_1<l<n}q_{il}(\varphi_n(t,l,z)-\varphi^*(t,l,z))\\
&+\sum_{l>n}q_{il}(1-\varphi^*(t,i,z))\leq\frac\varepsilon2+\sum_{l>n_1}q_{il}\leq \frac\varepsilon2+\frac\varepsilon2=\varepsilon.
\end{split}
\end{equation}
Thus, we deduce that $B_2^{(n)}(t,i,z)\rightrightarrows0$ in $t\in[0,T]$ as $n\to\infty$. We can have from Lemma~\ref{lem:conbound} that for all $x\in\mathds{R}^{N+1}$ satisfying $0\leq x\leq 2$, $0\leq\Phi(x)\leq C$ for some constant $C>0$. As for $j=0,1,\ldots,N$, $\varphi_n(t,i,z^j)\rightrightarrows\varphi^*(t,i,z^j)$ in $t\in[0,T]$ as $n\rightarrow\infty$, Lemma \ref{lem:lobndphistar} yields that there exists a constant $\delta>0$ such that $1\geq\varphi_n(t,i,z^j)\geq\varphi^*(t,i,z^j)\geq\delta>0$ for all $t\in[0,T]$. Further, there exists $\lambda^j_n(t)\in[0,1]$ such that $\varphi_n(t,i,z^j)=(1-\lambda^j_n(t))\varphi^*(t,i,z^j)+2\lambda^j_n(t)$. In turn,  $\lambda^j_n(t)=\frac{\varphi_n(t,i,z^j)-\varphi^*(t,i,z^j)}{2-\varphi^*(t,i,z^j)}$, and hence for all $j=0,1,\ldots,N$, $\lambda^j_n(t)\rightrightarrows0$ in $t\in[0,T]$ as $n\rightarrow\infty$. Similar to that in \eqref{concaveexpansion1}, we can derive that
\begin{align}\label{eq:infdiffer}
\Phi(x^{(n)}(t))\geq \Phi(x^*(t))\prod_{j=0}^N(1-\lambda^j_n(t))+\Lambda^{(n)}_1(t).
\end{align}
Similar to the first term in the r.h.s. of the inequality \eqref{eq:infdiffer}, every term in $\Lambda^{(n)}_1(t)$ above has $N+1$ multipliers and at least one of these multipliers is of the form $\lambda^j_n(t)$, while other multipliers are nonnegative and bounded by $1\vee C$. Due to the fact that $\lambda^j_n(t)\rightrightarrows0$ in $t\in[0,T]$ as $n\to\infty$, we have that $\Lambda^{(n)}_1(t)\rightrightarrows0$ in $t\in[0.T]$ as $n\to\infty$. Moreover, it follows from \eqref{eq:infdiffer} that
\begin{align}\label{eq:infdiffer2}
&\left(1-\prod_{j=0}^N(1-\lambda^j_n(t))\right)\Phi(x^*(t))-\Lambda^{(n)}_1(t)\geq\Phi(x^*(t))-\Phi(x^{(n)}(t))=-B_3^{(n)}(t,i,z).
\end{align}
It is not difficult to see that the l.h.s. of the inequality \eqref{eq:infdiffer2} tends to $0$ uniformly in $t\in[0,T]$ as $n\rightarrow\infty$.
On the other hand, there exists $\tilde{\lambda}^j_n(t)\in[0,1]$ such that $\varphi^*(t,i,z^j)=(1-\tilde{\lambda}^j_n(t))\varphi_n(t,i,z^j)+0\cdot\tilde{\lambda}^j_n(t)$, and in turn $\tilde{\lambda}^j_n(t)=\frac{\varphi_n(t,i,z^j)-\varphi^*(t,i,z^j)}{\varphi_n(t,i,z^j)}\rightrightarrows0$ in $t\in[0,T]$ as $n\to\infty$, since $\varphi_n(t,i,z^j)\geq\delta>0$. So that
\begin{align}\label{eq:infdiffer1}
&\left(1-\prod_{j=0}^N(1-\tilde{\lambda}^j_n(t))\right)\Phi(x^{(n)}(t))-\Lambda^{(n)}_2(t)\geq\Phi(x^{(n)}(t))-\Phi(x^{*}(t))=B_3^{(n)}(t,i,z),
\end{align}
where the form of $\Lambda^{(n)}_2(t)$ is similar to that of $\Lambda^{(n)}_1(t)$, but it is related to $\tilde{\lambda}^j_n(t)$ for $j=0,1,\ldots,N$.
As in \eqref{eq:infdiffer2}, the l.h.s. of the inequality~\eqref{eq:infdiffer1} tends to $0$ uniformly in $t\in[0,T]$ as $n\to\infty$.
Hence, it follows from \eqref{eq:infdiffer2} and \eqref{eq:infdiffer1} that $B_3^{(n)}(t,i,z)\rightrightarrows0$ in $t\in[0,T]$ as $n\rightarrow\infty$.
Thus, we proved that for $(i,z)\in\Zx_+\times{\cal S}$, $\frac{\partial\varphi_n(t,i,z)}{\partial t}\rightrightarrows\tilde{\varphi}(t,i,z)$ in $t\in[0,T]$
as $n\to\infty$.

We at last show that, for $(i,z)\in\Zx_+\times{\cal S}$, $\varphi^*(T,i,z)-\varphi^*(t,i,z)=\int_t^T\tilde{\varphi}(s,i,z)ds$ for $t\in[0,T]$.
For $n\in\Zx_+$, it follows from Proposition~\ref{prop:Vnmonotone00} that $\varphi_n(\cdot,i,z)\in C^1([0,T))\cap C([0,T])$ for $(i,z)\in D_n^0\times{\cal S}$. This implies that
\begin{equation}\label{differ}
\begin{split}
\varphi^*(T,i,z)-\varphi^*(t,i,z)&=\varphi^*(T,i,z)-\varphi^*(t,i,z)-(\varphi_n(T,i,z)-\varphi_n(t,i,z))\\
&\quad+\int_t^T\frac{\partial\varphi_n(s,i,z)}{\partial t}(s,i,z)ds.
\end{split}
\end{equation}
Lemma~\ref{lem:unfmconforphi} ensures that $\varphi(T,i,z)-\varphi(t,i,z)-(\varphi_n(T,i,z)-\varphi_n(t,i,z))\to0$ as $n\to\infty$.
From Lemma~\ref{lem:boundfordphi} and the uniform convergence of $\frac{\partial\varphi_n(t,i,z)}{\partial t}$ to $\tilde{\varphi}(t,i,z)$ in $t\in[0,T]$, it follows that $\tilde{\varphi}(t,i,z)$ is continuous in $t\in[0,T]$ and
$\int_t^T\frac{\partial\varphi_n(s,i,z)}{\partial t}ds\to\int_t^T\tilde{\varphi}(s,i,z)ds$ as $n\to\infty$.
Moreover, as $\varphi^*(T,i,z)-\varphi^*(t,i,z)=\int_t^T\tilde{\varphi}(s,i,z)ds$ for each $t\in[0,T]$, $\frac{\partial\varphi^*(t,i,z)}{\partial t}=\tilde{\varphi}(t,i,z)$ holds for all $t\in[0,T]$. Hence, $\varphi^*(t,i,z)$ is indeed a classical solution of the original DPE \eqref{eq:dpe3}. \hfill$\Box$\\

The verification argument for the case of countable state space $\Zx_+=\{1,2,\ldots\}$ is presented in the next key proposition. Before it, we provide some mild conditions on model coefficients:
\begin{itemize}
  \item[({C.1})] There exist positive constants $c_1$, $c_2$, $\delta$ and $K$ such that $c_1\|\xi\|^2\leq\xi^\top\sigma(i)\sigma(i)^\top\xi\leq c_2\|\xi\|^2$ for all $\xi\in\R^N$ and $i\in\Zx_+$, $\delta\leq\lambda(i,z)\leq K$ for all $(i,z)\in\Zx_+\times{\cal S}$, and $r(i)+\|\mu(i)\|\leq K$ for all $i\in\Zx_+$.
\end{itemize}
The first condition on $\sigma(i)$ is actually related to the uniformly elliptic property of the volatility matrix $\sigma(i)$ of stocks.
\begin{proposition}\label{prop:verivalue}
Let the condition {\rm(C.1)} hold. Let $\varphi^*(t,i,z)$ with $(t,i,z)\in[0,T]\times\Zx_+\times{\cal S}$ be given by \eqref{eq:varphistar}. Then, for all $(t,i,z)\in[0,T]\times\Zx_+\times{\cal S}$,
\begin{align}\label{veriphistar}
\varphi^*(t,i,z)=\inf_{\tilde{\pi}\in\tilde{\cal U}}\Ex_{t,i,z}^{\tilde{\pi},\theta}\left[\exp\left(\frac{\theta}{2}\int_t^TL(\tilde{\pi}(s);Y(s),Z(s))ds\right)\right].
\end{align}
\end{proposition}

\begin{proof}
From Proposition~\ref{prop:verithemfinite} and Lemma~\ref{lem:jn=tildeJn}, it follows that, for $n\in\Zx_+$,
\begin{align*}
\varphi_n(t,i,z)=&\inf_{\tilde{\pi}\in\tilde{\cal U}_n}\Ex_{t,i,z}^{\tilde{\pi},\theta}\left[\exp\left(\frac{\theta}{2}\int_t^TL(\tilde{\pi}(s);Y^{(n)}(s),Z(s))ds\right)\right]\nonumber\\
=&\inf_{\tilde{\pi}\in\tilde{\cal U}_n}\Ex_{t,i,z}^{\tilde{\pi},\theta}\left[\exp\left(\frac{\theta}{2}\int_t^{T\wedge \tau^t_n}L(\tilde{\pi}(s);Y(s),Z(s))ds\right)\right].
\end{align*}
Then, for any $\varepsilon>0$, there exists $\tilde{\pi}^\varepsilon\in\tilde{\mathcal{U}}_n$ such that
\begin{align}\label{eq:varphin+epsilon}
\varphi_n(t,i,z)+\varepsilon>\Ex_{t,i,z}^{\pi^\varepsilon,\theta}\left[\exp\left(\frac{\theta}{2}\int_t^{T\wedge \tau^t_n}L(\tilde{\pi}^\varepsilon(s);Y(s),Z(s))ds\right)\right].
\end{align}
Define $\hat{\pi}^\varepsilon(t):=\tilde{\pi}^\varepsilon(t)\mathds{1}_{\{t\leq\tau_n\}}$ for $t\in[0,T]$. Then, it holds that $\hat{\pi}^\epsilon\in\tilde{\mathcal{U}}$, and $\Gam^{\hat{\pi}^\varepsilon,\theta}(t,T)=\Gam^{\tilde{\pi}^\varepsilon,\theta}(t,T\wedge\tau^t_n)$ for $t\in[0,T]$. Also note that $L(0,i,z)=-r(i)\leq0$ for all $(i,z)\in\Zx_+\times{\cal S}$. Then, the inequality~\eqref{eq:varphin+epsilon} continues that
\begin{align}
\varphi_n(t,i,z)+\varepsilon>&\Ex_{t,i,z}^{\tilde{\pi}^\varepsilon,\theta}\left[\exp\left(\frac{\theta}{2}\int_t^{T\wedge \tau^t_n}L(\tilde{\pi}^\varepsilon(s);Y(s),Z(s))ds\right)\right]\nonumber\\
=&\Ex_{t,i,z}^{\hat{\pi}^\varepsilon,\theta}\left[\exp\left(\frac{\theta}{2}\int_t^{T\wedge \tau^t_n}L(\hat{\pi}^\varepsilon(s);Y(s),Z(s))ds\right)\right]\nonumber\\
\geq&\Ex_{t,i,z}^{\hat{\pi}^\varepsilon,\theta}\left[\exp\left(\frac{\theta}{2}\int_t^{T}L(\hat{\pi}^\varepsilon(s);Y(s),Z(s))ds\right)\right]\nonumber\\
\geq&\inf_{\tilde{\pi}\in\tilde{\cal U}}\Ex_{t,i,z}^{\tilde{\pi},\theta}\left[\exp\left(\frac{\theta}{2}\int_t^TL(\tilde{\pi}^\varepsilon(s);Y(s),Z(s))ds\right)\right].
\end{align}
By passing $n\to\infty$ and then $\varepsilon\to0$, we get
\begin{align}\label{phistaroninf}
\varphi^*(t,i,z)\geq\inf_{\tilde{\pi}\in\tilde{\cal U}}\Ex_{t,i,z}^{\tilde{\pi},\theta}\left[\exp\left(\frac{\theta}{2}\int_t^TL(\tilde{\pi}(s);Y(s),Z(s))ds\right)\right].
\end{align}
On the other hand, using Theorem~\ref{thm:existD} and Proposition~\ref{prop:verithemfinite}, $\varphi^*(t,i,z)$ is strictly positive and $\varphi^*(t,i,z)\leq\varphi_n(t,i,z)\leq1$ for all $n\geq1$. Then, under the condition (C.1), by applying a similar argument of the proof of \eqref{eq:itoveri},  we have that, for any $\tilde{\pi}\in\tilde{\cal U}$,
\begin{align*}
&\Ex_{t,i,z}^{\tilde{\pi},\theta}\left[\varphi^*(T,Y(T),Z(T))\exp\left(\frac{\theta}{2}\int_t^TL(\tilde{\pi}(u);Y(u),Z(u))du\right)\right]\geq\varphi^*(t,i,z).
\end{align*}
Because $\varphi(T,i,z)=1$ for all $(i,z)\in\Zx_+\times{\cal S}$, we deduce that
\begin{align}\label{infonphistar}
\inf_{\tilde{\pi}\in\tilde{\cal U}}\Ex_{t,i,z}^{\tilde{\pi},\theta}\left[\exp\left(\frac{\theta}{2}\int_t^TL(\tilde{\pi}(s);Y(s),Z(s))ds\right)\right]\geq\varphi^*(t,i,z).
\end{align}
The equality \eqref{veriphistar} therefore follows by combining \eqref{phistaroninf} and \eqref{infonphistar}, and the validity of the proposition is checked.
\end{proof}

Similar to that in Proposition~\ref{prop:verithemfinite}, we can construct a candidate optimal $\Gx$-predictable feedback strategy $\tilde{\pi}^*$ by, for $t\in[0,T]$,
\begin{align}\label{eq:optimaltildepis}
\tilde{\pi}^*(t)&:={\rm diag}\left((1-Z_j(t-))_{j=1}^N\right)\nonumber\\
&\quad\times\argmin_{\pi\in U}\tilde{H}\left(\pi;Y(t-),Z(t-),(\varphi^*(t,Y(t-),Z^j(t-));\ j=0,1,\ldots,N)\right).
\end{align}
We first prove that $\tilde{\pi}^*$ can be characterized as an approximation limit by a sequence of well defined admissible strategies.
\begin{lemma}\label{lem:approxpistar}
Let the condition {\rm(C.1)} hold. There exists a sequence of strategies $(\tilde{\pi}^{(n,*)})_{n\in\Zx_+}\subset\tilde{\mathcal{U}}$ such that $\lim_{n\to\infty}\tilde{\pi}^{(n,*)}(t)=\tilde{\pi}^*(t)$ for $t\in[0,T]$, $\Px$-a.s., and further $\lim_{n\to\infty}J(\tilde{\pi}^{(n,*)};t,i,z)=\varphi^*(t,i,z)$ for $(t,i,z)\in[0,T]\times\Zx_+\times{\cal S}$, $\Px$-a.s. Here, the objective functional $J$ is defined in~\eqref{eq:J}.
\end{lemma}

\begin{proof}
For fixed $(i,z,x)\in\Zx_+\times\mathcal{S}\times(0,\infty)^{N+1}$, we have that $\tilde{H}\left(\pi;i,z,x\right)$ is strictly concave w.r.t. $\pi\in U$, and hence $\Phi(i,z,x):=\argmin_{\pi\in U}\tilde{H}\left(\pi;i,z,x\right)$ is well defined. Note that $\Phi(i,z,\cdot)$ maps $(0,\infty)^{N+1}$ to $U$ and satisfies the first-order condition $\frac{\partial\tilde{H}}{\partial\pi_j}\left(\Phi(i,z,x);i,z,x\right)=0$ for $j=1,\ldots,N$.
Then, Implicit Function Theorem yields that $\Phi(i,z,x)$ is continuous in $x$. Let $x^{(n)}(t):=(\varphi_n(t,Y^{(n)}(t-),Z^j(t-));\ j=0,1,\ldots,N)$. It follows from Proposition~\ref{prop:verithemfinite} and Lemma~\ref{lem:jn=tildeJn} that, for $t\in[0,T]$,
\begin{equation}
\tilde{\pi}^{(n,*)}(t):={\rm diag}((1-Z_j(t-))_{j=1}^N)\Phi(Y(t-),Z(t-),x^{(n)}(t))\mathds{1}_{\{t\leq\tau_n\}}\nonumber
\end{equation}
belongs to $\tilde{\mathcal{U}}_n\cap\tilde{\mathcal{U}}$, and further it satisfies that
\begin{align}
\varphi_n(t,i,z)&=\Ex_{t,i,z}^{{\tilde\pi^{(n,*)}},\theta}\left[\exp\left(\frac{\theta}{2}\int_t^{T\wedge\tau^t_n}L(\tilde{\pi}^{(n,*)}(s);Y(s),Z(s))ds\right)\right].
\end{align}
Lemma~\ref{lem:unfmconforphi} gives that $\lim_{n\to\infty}\|x^{(n)}(t)-x^*(t)\|=0$ for $t\in[0,T]$, $\Px$-a.s., where $x^*(t):=(\varphi^*(t,Y(t-),Z^j(t-));\ j=0,1,\ldots,N)$. We define the predictable process $\tilde{\pi}^*(t):={\rm diag}((1-Z_j(t-))_{j=1}^N)\Phi(Y(t-),Z(t-),x^*(t))$ for $t\in[0,T]$. By Lemma~\ref{lem:lobndphistar} and the continuity of $\Phi(i,z,\cdot)$, we obtain  $\lim_{n\to\infty}\tilde{\pi}^{(n,*)}(t)=\tilde{\pi}^*(t)$ for $t\in[0,T]$, a.s. Moreover, it holds that
\begin{align*}
&J(\tilde{\pi}^{(n,*)};t,i,z)=\Ex_{t,i,z}^{{\tilde\pi^{(n,*)}},\theta}\left[\exp\left(\frac{\theta}{2}\int_t^TL(\tilde\pi^{(n,*)}(s);Y(s),Z(s))ds\right)\right]\nonumber\\
&\qquad=\Ex_{t,i,z}^{{\tilde\pi^{(n,*)}},\theta}\left[\exp\left(\frac{\theta}{2}\int_t^{T\wedge\tau^t_n}L(\tilde\pi^{(n,*)}(s);Y(s),Z(s))ds+\frac{\theta}{2}\int_{T\wedge\tau^t_n}^TL(0;Y(s),Z(s))ds\right)\right]\nonumber\\
&\qquad\leq\Ex_{t,i,z}^{{\tilde\pi^{(n,*)}},\theta}\left[\exp\left(\frac{\theta}{2}\int_t^{T\wedge\tau^t_n}L(\tilde\pi^{(n,*)}(s);Y(s),Z(s))ds\right)\right]=\varphi_n(t,i,z).
\end{align*}
Proposition~\ref{prop:verivalue} then yields that $\varphi^*(t,i,z)\leq J(\tilde{\pi}^{(n,*)};t,i,z)\leq\varphi_n(t,i,z)$ for $n\in\Zx_+$. This verifies that $\lim_{n\to\infty}J(\tilde{\pi}^{(n,*)};t,i,z)=\varphi^*(t,i,z)$ for $(t,i,z)\in[0,T]\times\Zx_+\times{\cal S}$, a.s. using Lemma~\ref{lem:unfmconforphi}.
\end{proof}

\begin{proposition}\label{prop:admiss}
Let the condition {\rm(C.1)} hold. Then, the optimal feedback strategy $\tilde{\pi}^*$ given by \eqref{eq:optimaltildepis} is admissible, i.e., $\tilde{\pi}^*\in\tilde{\cal U}$.
\end{proposition}

\begin{proof}
Under the condition {\rm(C.1)}, it is not difficult to verify that there exists a constant $C>0$ such that $L(\pi;i,z)\geq-C$ for all $(\pi,i,z)\in U\times\Zx_+\times\mathcal{S}$. Thanks to Proposition~\ref{prop:verivalue}, we have that
\begin{align*}
\varphi^*(t,i,z)&=\inf_{\tilde{\pi}\in\tilde{\cal U}}\Ex_{t,i,z}^{\tilde{\pi},\theta}\left[\exp\left(\frac{\theta}{2}\int_t^TL(\tilde{\pi}(s);Y(s),Z(s))ds\right)\right]\\
&\geq\inf_{\tilde{\pi}\in\tilde{\cal U}}\Ex_{t,i,z}^{\tilde{\pi},\theta}\left[\exp\left(-\frac{\theta}{2}\int_t^TCds\right)\right]=\exp\left(-\frac\theta2C(T-t)\right),\nonumber
\end{align*}
for $(t,i,z)\in[0,T]\times{\Zx_+}\times{\cal S}$. Hence, for $t\in[0,T]$,
\begin{align}\label{eq:xstarbound}
x^*(t)=(\varphi^*(t,Y(t-),Z^j(t-));\ j=0,1,\ldots,N))\in[e^{-\frac\theta2C(T-t)},1]^{N+1}.
\end{align}
The continuity of $\Phi(i,z,x):=\argmin_{\pi\in U}\tilde{H}\left(\pi;i,z,x\right)$ gives that $\tilde{\pi}^*(t)$ for $t\in[0,T]$ is uniformly bounded by some constant $C_1>0$. Moreover, the first-order condition yields that, for all $j=1,\ldots,N$, if $Z_j(t-)=0$,
\begin{align}\label{eq:pistarbelow2}
(1-\tilde{\pi}^*_j(t))^{-\frac\theta2-1}
=&\Bigg[(\mu_j(Y(t-))-r(Y(t-)))-\frac\theta2\left(1+\frac\theta2\right)\sum_{i=1}^N(\sigma(Y(t-))^\top\sigma(Y(t-)))_{ji}\tilde{\pi}^*_i(t)\nonumber\\
&\quad+\frac\theta2\lambda_j(Y(t-),Z(t-))\Bigg]
\frac{\varphi^*(t,Y(t-),Z(t-))}{\lambda_j(Y(t-),Z(t-))\varphi^*(t,Y(t-),Z^j(t-))}\nonumber\\
\leq& C_2,
\end{align}
where we used the condition (C.1) and \eqref{eq:xstarbound}. Note that $\tilde{\pi}^*_j(t)=0$ if $Z_j(t-)=1$, then $\tilde{\pi}^*$ is also uniformly bounded away from $1$. This implies that the generalized Novikov's condition holds in the countably infinite state case, and hence $\tilde{\pi}^*$ is admissible.
\end{proof}
The above verification results (Proposition~\ref{prop:verivalue} and Proposition~\ref{prop:admiss}) can be seen as a uniqueness result for the dynamic programming equation. Under the condition (C.1), we can also establish an error estimate on the approximation of the sequence of strategies $\tilde{\pi}^{(n,*)}$ to the optimal strategy $\pi^{*}$ in terms of the objective functional $J$ (see~\eqref{eq:J}), which is given by
\begin{lemma}\label{lem:errorestimate}
Let $n\in\Zx_+$. Under the condition {\rm(C.1)}, for $(t,i,z)\in[0,T]\times D_n\times{\cal S}$, there exists a constant $C>0$ which is independent of $n$ such that
\begin{align*}
\left|J(\tilde{\pi}^{(n,*)};t,i,z)-J(\tilde{\pi}^{(*)};t,i,z)\right|\leq C\left(1-\sum_{j=1}^na^{(n)}_{ij}(T-t)\right).
\end{align*}
Here $a^{(n)}_{ij}(T-t)=\delta_{ij}+(T-t)q_{ij}+\sum_{k=1}^\infty\sum_{1\leq l_1,\ldots,l_k\leq n}\frac{(T-t)^{k+1}}{(k+1)!}q_{il_1}q_{l_1l_2}\cdots q_{l_k j}$.
\end{lemma}

\begin{proof}
By Proposition 4.5, $J(\tilde{\pi}^{(n,*)};t,i,z)\to \varphi^*(t,i,z)=J(\tilde{\pi}^*;t,i,z)$ as $n\to\infty$. On the other hand, it can be verified that there exists constants $\gamma\in(0,1)$ and $C_1>0$ such that $\tilde{\pi}^*(t)\in[-C_1,1-\gamma]^N$ for all $t\in[0,T]$, a.s. Then, using \eqref{eq:L0}, it follows that $L(\tilde{\pi}^*(t);Y(t),Z(t))\leq C_2$, a.s. for $t\in[0,T]$. Here $C_2$ is a positive constant. Therefore, by noting $\tilde{\pi}^*\in\tilde{\mathcal{U}}_n$, we have that
\begin{align*}
\varphi^*(t,i,z)&=\Ex_{t,i,z}^{\tilde{\pi}^*,\theta}\left[\exp\left(\frac{\theta}{2}\int_t^TL(\tilde{\pi}^*(s);Y(s),Z(s))ds\right)\right]\notag\\
&\geq\Ex_{t,i,z}^{\tilde{\pi}^*,\theta}\left[\exp\left(\frac{\theta}{2}\int_t^TL(\tilde{\pi}^*(s);Y(s),Z(s))ds\right)\mathbf{1}_{\{\tau^t_n>T\}}\right]\notag\\
&=\Ex_{t,i,z}^{\tilde{\pi}^*,\theta}\left[\exp\left(\frac{\theta}{2}\int_t^{T\wedge\tau^t_n}L(\tilde{\pi}^*(s);Y(s),Z(s))ds\right)\mathbf{1}_{\{\tau^t_n>T\}}\right]\notag\\
&=\Ex_{t,i,z}^{\tilde{\pi}^*,\theta}\left[\exp\left(\frac{\theta}{2}\int_t^{T\wedge\tau^t_n}L(\tilde{\pi}^*(s);Y(s),Z(s))ds\right)\right]\notag\\
&\qquad-\Ex_{t,i,z}^{\tilde{\pi}^*,\theta}\left[\exp\left(\frac{\theta}{2}\int_t^{T\wedge\tau^t_n}L(\tilde{\pi}^*(s);Y(s),Z(s))ds\right)\mathbf{1}_{\{\tau^t_n\leq T\}}\right]\notag\\
&\geq\varphi_n(t,i,z)-\Ex_{t,i,z}^{\tilde{\pi}^*,\theta}\left[e^{\frac{\theta C_2}{2}(T\wedge\tau^t_n-t)}\mathbf{1}_{\{\tau^t_n\leq T\}}\right]\notag\\
&\geq\varphi_n(t,i,z)-C_3\mathbb{P}_{t,i,z}^{\tilde{\pi}^*,\theta}(\tau^t_n\leq T),
\end{align*}
where $C_3:=e^{\frac{\theta C_2T}{2}}$ and $\varphi_n(t,i,z)$ is defined in Proposition~\ref{prop:Vnmonotone00}. Using the given inequality $\varphi^*(t,i,z)\leq J(\tilde{\pi}^{(n,*)};t,i,z)\leq\varphi_n(t,i,z)$ in the proof of Lemma \ref{lem:approxpistar}, under the condition (C.1), we arrive at
\begin{align*}
\left|J(\tilde{\pi}^{(n,*)};t,i,z)-J(\tilde{\pi}^{(*)};t,i,z)\right|&=J(\tilde{\pi}^{(n,*)};t,i,z)-\varphi^*(t,i,z)\leq\varphi_n(t,i,z)-\varphi^*(t,i,z)\nonumber\\
&\leq C_3\mathbb{P}_{t,i,z}^{\tilde{\pi}^*,\theta}(\tau^t_n\leq T).
\end{align*}
Note that, by Proposition 4.5, $Y$ is also a Markov chain with the generator $Q=(q_{ij})$ under $\mathbb{P}_{t,i,z}^{\tilde{\pi}^*,\theta}$. Then $\mathbb{P}_{t,i,z}^{\tilde{\pi}^*,\theta}(\tau^t_n\leq T)\to0$ as $n\to\infty$. On the other hand, $\tau^t_n$ is the absorption time of $(Y^{(n)}(s))_{s\in[t,T]}$ whose generator is given as $A_n$ given by \eqref{eq:An}. Hence, using Section 11.2.3 in Chapter 11 in ~\cite{BieRut04}, we also have that $\Px_{t,i,z}^{\tilde{\pi}^*,\theta}(\tau^t_n\leq T)=1-\sum_{j=1}^na^{(n)}_{ij}(T-t)$. This completes the proof.
\end{proof}

We next provide an example in which the error estimate $1-\sum_{j=1}^na_{ij}^{(n)}(T-t)$ in Lemma~\ref{lem:errorestimate} admits a closed form representation. Let us consider the following specific generator given by
\begin{align*}%\label{eq:example}
Q=\left[\begin{matrix}
     -1 & \frac12 & \frac14 & \dots & \frac1{2^{n-1}} & \frac1{2^n} & \dots \\
     \frac12 & -1 & \frac14 & \dots & \frac1{2^{n-1}} & \frac1{2^n} & \dots\\
     \frac12 & \frac14 & -1 & \dots & \frac1{2^{n-1}} & \frac1{2^n} & \dots \\
         \vdots   & \vdots  & \vdots & \vdots & \vdots     \\
     \frac12 & \frac14 & \frac18 &\dots & \frac1{2^{n-1}} & -1 & \dots\\
     \vdots   & \vdots  & \vdots & \vdots  & \vdots & \vdots   \\
\end{matrix}\right].\notag
\end{align*}
Then, for any $l\leq n$, $\sum_{j=1}^nq_{lj}=\sum_{j=1}^{n-1}\frac1{2^j}-1=\frac{-1}{2^{n-1}}$. Therefore, for any $i\leq n$,
\begin{align*}
\sum_{j=1}^na^{(n)}_{ij}(T-t)&=\sum_{k=0}^\infty\frac{(T-t)^k}{k!}\left(\frac{-1}{2^{n-1}}\right)^k=e^{-\frac{T-t}{2^{n-1}}}.
\end{align*}
It follows that, for $(t,i,z)\in[0,T]\times D_n\times{\cal S}$, we have the explicit error estimate
\begin{eqnarray*}
\left|J(\tilde{\pi}^{(n,*)};t,i,z)-J(\tilde{\pi}^{(*)};t,i,z)\right|\leq C\left(1-e^{-\frac{T-t}{2^{n-1}}}\right),
\end{eqnarray*}
where $C>0$ is independent of $n$.

\begin{remark}\label{rem:qijt}
It is also worth mentioning here that our method used in the paper can be applied to treat the case where the regime-switching process $Y$ is a time-inhomogeneous Markov chain with a time-dependent generator given by $Q(t)=(q_{ij}(t))_{i,j\in\Zx_+}$ for $t\in[0,T]$. Here, for $t\in[0,T]$, $q_{ii}(t)\leq0$ for $i\in\Zx_+$, $q_{ij}(t)\geq0$ for $i\neq j$, and $\sum_{j=1}^{\infty}q_{ij}(t)=0$ for $i\in\Zx_+$ (i.e., $\sum_{j\neq i}q_{ij}(t)=-q_{ii}(t)$ for $i\in\Zx_+$). Also for $i,j\in\Zx_+$, $t\to q_{ij}(t)$ is continuous on $[0,T]$, and the infinite summation $\sum_{j\in\Zx_+}q_{ij}(t)$ is uniformly convergent in $t\in[0,T]$.\\
\end{remark}

\noindent
\textbf{Acknowledgements}: L. Bo is supported by Natural Science Foundation of China under grant 11471254 and the Key Research Program of Frontier Sciences of the Chinese Academy of Science under grant QYZDB-SSW-SYS009. X. Yu is supported by the Hong Kong Early Career Scheme under grant 25302116. The authors would like to thank two anonymous referees for the careful reading and helpful comments to improve the presentation of this paper.

\appendix
\section{Auxiliary Lemmas}\label{app:proof1}
\renewcommand\theequation{A.\arabic{equation}}
\setcounter{equation}{0}

\begin{lemma}\label{lem:conconver}
Let the function $\Phi(x):[0,\infty)^{N+1}\to\R$ be concave in every component of $x$. Assume that there exists $\overline{x}$, $x^*$, $\underline{x}\in[0,\infty)^{N+1}$ such that $\underline{x}\ll x^*\ll\overline{x}$. Let $\{x^{(n)}\}_{n\geq1}\subset[0,\infty)^{N+1}$ satisfy $x^*\leq x^{(n)}$ for $n\geq1$ and $\lim_{n\to\infty}x^{(n)}=x^*$. Then $\lim_{n\to\infty}\Phi(x^{(n)})=\Phi(x^*)$.
%\begin{equation}
%\lim_{n\to\infty}\Phi(x^{(n)})=\Phi(x^*).\nonumber
%\end{equation}
\end{lemma}

\begin{proof} Due to the given conditions in the lemma, there exists $n_0\geq1$ such that $x^*\leq x^{(n)}\leq\overline{x}$ for all $n\geq n_0$. For each $n\geq n_0$, there exists a vector $\lambda^{(n)}\in[0,1]^{N+1}$ satisfying $\lim_{n\to\infty}\lambda^{(n)}=0$ such that $x^{(n)}_k=\lambda^{(n)}_k\overline{x}_k+(1-\lambda^{(n)}_k)x^*_k$ for $k=1,\ldots,N+1$. Therefore, it follows that
\begin{align}\label{concaveexpansion1}
\Phi(x^{(n)})=&\Phi\Big(\lambda^{(n)}_1\overline{x}_1+(1-\lambda^{(n)}_1)x^*_1,\lambda^{(n)}_2\overline{x}_2\nonumber\\
&+(1-\lambda^{(n)}_2)x^*_2,\ldots,\lambda^{(k)}_{N+1}\overline{x}_{N+1}+(1-\lambda^{(n)}_{N+1})x^*_{N+1}\Big)\nonumber\\
\geq&\lambda^{(n)}_1\Phi\Big(\overline{x}_1,\lambda^{(n)}_2\overline{x}_2+(1-\lambda^{(n)}_2)x^*_2,\ldots,\lambda^{(k)}_{N+1}\overline{x}_{N+1}+(1-\lambda^{(n)}_{N+1})x^*_{N+1}\Big)\nonumber\\
&+(1-\lambda^{(n)}_1)\Phi\Big(x^*_1,\lambda^{(n)}_2\overline{x}_2+(1-\lambda^{(n)}_2)x^*_2,\ldots,\lambda^{(k)}_{N+1}\overline{x}_{N+1}+(1-\lambda^{(n)}_{N+1})x^*_{N+1}\Big)\nonumber\\
\geq&\lambda^{(n)}_1\lambda^{(n)}_2\Phi\Big(\overline{x}_1,\overline{x}_2,\ldots,\lambda^{(n)}_{N+1}\overline{x}_{N+1}+(1-\lambda^{(n)}_{N+1})x^*_{N+1}\Big)\nonumber\\
&+\lambda^{(n)}_1(1-\lambda^{(n)}_2)\Phi\Big(\overline{x}_1,x^*_2,\ldots,\lambda^{(k)}_{N+1}\overline{x}_{N+1}+(1-\lambda^{(n)}_{N+1})x^*_{N+1}\Big)\nonumber\\
&+(1-\lambda^{(n)}_1)\lambda^{(n)}_2\Phi\Big(\overline{x}_1,x^*_2,\ldots,\lambda^{(k)}_{N+1}\overline{x}_{N+1}+(1-\lambda^{(n)}_{N+1})x^*_{N+1}\Big)\nonumber\\
&+(1-\lambda^{(n)}_1)(1-\lambda^{(n)}_2)\Phi\Big(x^*_1,x^*_2,\ldots,\lambda^{(k)}_{N+1}\overline{x}_{N+1}+(1-\lambda^{(n)}_{N+1})x^*_{N+1}\Big)\nonumber\\
\geq&\Phi(x^*)\prod_{k=1}^{N+1}(1-\lambda_{k}^{(n)})+\Sigma^{(n)}_1.
\end{align}

We observe that every term in $\Sigma^{(n)}_1$ above has one or more multipliers which is of the form $\lambda^{(n)}_k$ for $k=1,\ldots,N+1$. As $\lim_{n\to\infty}\lambda^{(n)}_k=0$ for $k=1,\ldots,N+1$ and hence $\Sigma^{(n)}_1\to0$ as $n\to\infty$. It follows from \eqref{concaveexpansion1} that $\liminf_{n\rightarrow\infty}\Phi(x^{(n)})\geq\Phi(x^*)$. Similarly, as $x^{(n)}\geq x^*\gg \underline{x}$ for all $n\in\N$, there exists a vector $\tilde{\lambda}^{(n)}\in[0,1]^{N+1}$ satisfying $\lim_{n\to\infty}\tilde{\lambda}^{(n)}=0$ such that $x^*_k=\tilde{\lambda}^{(n)}_k\underline{x}_k+(1-\tilde{\lambda}^{(n)}_k)x^{(n)}_k$ for $k=1,\ldots,N+1$. Using the similar argument in the proof of \eqref{concaveexpansion1}, we deduce that
\begin{align}\label{concaveexpansion2}
\Phi(x^*)\geq\Phi(x^{(n)})\prod_{k=1}^{N+1}(1-\tilde{\lambda}_{k}^{(n)})+\Sigma^{(n)}_2,
\end{align}
where every term in $\Sigma^{(n)}_2$ above has one or more multipliers which is of the form $\lambda^{(n)}_k$, $k=1,\ldots,N+1$. The inequality \eqref{concaveexpansion2} gives that $\Phi(x^*)\geq\limsup_{n\to\infty}\Phi(x^{(n)})$. Putting the above two inequalities together, we obtain $\lim_{n\to\infty}\Phi(x^{(n)})=\Phi(x^*)$, which completes the proof.
\end{proof}

\begin{lemma}\label{lem:conbound}
Let the function $\Phi(x):[0,\infty)^{N+1}\to[0,\infty)$ be concave in every component of $x$. Then, for any $\alpha,\beta\in[0,\infty)^{N+1}$ satisfying $\alpha\leq\beta$, there exists a constant $C=C(\alpha,\beta)>0$ such that $0\leq\Phi(x)\leq C$ for all $\alpha\leq x\leq \beta$.
\end{lemma}

\begin{proof} For any $\alpha\leq x\leq \beta$ where $\alpha,\beta\in[0,\infty)^{N+1}$, there exists a vector $\nu\in(0,\infty)^{N+1}$ such that $\beta\ll\nu$. This implies that
there exists $\lambda\in[0,1]^{N+1}$ such that $\beta_k=\lambda_k x_k+(1-\lambda_k)\nu_k$, $k=1,\ldots,N+1$. As $\alpha\leq x\leq \beta\ll\nu$,
there exists $\delta>0$ such that $1\geq\lambda_k=\frac{\nu_k-\beta_k}{\nu_k-x_k}\geq\frac{\nu_k-\beta_k}{\nu_k-\alpha_k}\geq\delta$. Using the concave property of $\Phi(x)$, we have that
\begin{equation}\label{conbound1}
\begin{split}
\Phi(\beta)=&\Phi(\lambda_1 x_1+(1-\lambda_1)\nu_1,\lambda_2x_2+(1-\lambda_2)\nu_2,\ldots,\lambda_{N+1}x_{N+1}+(1-\lambda_{N+1})\nu_{N+1})\\
\geq&\Phi(x)\prod_{k=1}^{N+1}\lambda_k\\
&+\sum_{\substack{1\leq j_1<j_2<\ldots<j_k<N+1\\1\leq k\leq N+1}}(1-\lambda_{j_1})\times\cdots\times(1-\lambda_{j_k})\lambda_{j_{k+1}}\times\cdots\times\lambda_{j_N}\Phi(C_{j_1\ldots j_k})\\
\geq&\delta^{N+1}\Phi(x),
\end{split}
\end{equation}
for some $C_{j_1\ldots j_k}\in[0,\infty)^{N+1}$, and $\{j_{k+1},\ldots,j_{N+1}\}=\{1,\ldots,N+1\}\setminus\{j_1,\ldots,j_k\}$. We therefore have shown that the claim of the lemma holds.
\end{proof}

\end{document}